\documentclass[11pt]{article}
\usepackage{graphicx} 
\usepackage{fullpage}
\usepackage{xcolor}
\usepackage{multirow}
\usepackage{booktabs}
\usepackage{hyperref}
\usepackage{amsmath,amsthm}
\usepackage{amssymb}
\usepackage[ruled, noend]{algorithm2e}
\usepackage{mathtools}
\usepackage{tcolorbox}
\usepackage{enumitem}
\usepackage[T1]{fontenc}
\usepackage{libertine}
\usepackage{setspace}
\usepackage{subcaption}
\usepackage{ulem}
\SetCommentSty{mytcpstyle}

\def\eps{\epsilon}
\theoremstyle{definition}
\newtheorem{question}{Question}
\newtheorem{theorem}{Theorem}[section]
\newtheorem{observation}{Observation}
\newtheorem{definition}[theorem]{Definition}
\newtheorem{lemma}[theorem]{Lemma}
\newtheorem{corollary}[theorem]{Corollary}
\newtheorem{remark}[theorem]{Remark}

\newcommand{\parent}{\textnormal{parent}}
\newcommand{\opt}{\textnormal{OPT}}
\newcommand{\sol}{\textnormal{SOL}}
\newcommand{\vx}{\boldsymbol{x}}

\newcommand{\polylog}{\operatorname{polylog}}

\newcommand{\Min}{\textsc{Min}}
\newcommand{\Max}{\textsc{Max}}
\newcommand{\argmin}{\textnormal{argmin}}
\newcommand{\argmax}{\textnormal{argmax}}

\definecolor{mygreen}{RGB}{10,150,110}
\definecolor{myred}{RGB}{150,10,20}
\hypersetup{
     colorlinks=true,
     citecolor= mygreen,
     linkcolor= myred
}
\providecommand{\email}[1]{\href{mailto:#1}{\nolinkurl{#1}\xspace}}

\title{Streaming Algorithms for Network Design}
\author{Chandra Chekuri\thanks{University of Illinois at Urbana-Champaign: \email{chekuri@illinois.edu}. Supported
in part by NSF grant CCF-2402667.} 
\and Rhea Jain\thanks{University of Illinois at Urbana-Champaign: \email{rheaj3@illinois.edu}. Supported
in part by NSF grant CCF-2402667.} 
\and Sepideh Mahabadi\thanks{Microsoft Research: \email{smahabadi@microsoft.com}.} 
\and Ali Vakilian\thanks{Toyota Technological Institute at Chicago (TTIC): \email{vakilian@ttic.edu}.}}
\date{}
\begin{document}

\maketitle

\begin{abstract}
We consider the Survivable Network Design problem (SNDP) in the single-pass insertion-only streaming model. The input to SNDP is an edge-weighted graph $G = (V, E)$ and an integer connectivity requirement $r(uv)$ for each $u, v \in V$. The objective is to find a minimum-weight subgraph $H \subseteq G$ such that, for every pair of vertices $u, v \in V$, $u$ and $v$ are $r(uv)$-edge/vertex-connected. Recent work by Jin {\em et al.}~\cite{JKMV24} obtained approximation algorithms for edge-connectivity augmentation, and via that, also derived algorithms for edge-connectivity SNDP (EC-SNDP). In this work we consider \emph{vertex-connectivity} setting (VC-SNDP) and obtain several results for it as well as improved results for EC-SNDP.
    \begin{itemize}
        \item We provide a general framework for solving connectivity problems including SNDP and others in streaming; this is based on a connection to \emph{fault-tolerant spanners}. For VC-SNDP, we provide an $O(tk)$-approximation in $\tilde O(k^{1-1/t}n^{1 + 1/t})$ space, where $k$ is the maximum connectivity requirement, assuming an exact algorithm at the end of the stream. Using a refined LP-based analysis, we provide an $O(\beta t)$-approximation where $\beta$ is the integrality gap of the natural cut-based LP relaxation. These are the first approximation algorithms in the streaming model for VC-SNDP. When applied to the EC-SNDP, our framework provides an $O(t)$-approximation in $\tilde O(k^{1/2-1/(2t)}n^{1 + 1/t} + kn)$ space, improving the $O(t \log k)$-approximation of \cite{JKMV24} using $\tilde{O}(kn^{1+1/t})$ space; this also extends to element-connectivity SNDP. 
        \item We consider vertex connectivity-augmentation in the link-arrival model.  The input is a $k$-vertex-connected spanning subgraph $G$, and additional weighted links $L$ arrive in the stream; the goal is to store the minimum-weight set of links such that $G \cup L$ is $(k+1)$-vertex-connected. We obtain constant-factor approximations in near-linear space for $k = 1, 2$. Our result for $k=2$ is based on using the SPQR tree, a novel application for this well-known representation of $2$-connected graphs.
    \end{itemize}
\end{abstract}

\newpage
\section{Introduction}

Network design is a classical area of research in combinatorial optimization that has been instrumental in the development and advancement of several important algorithmic techniques. Moreover, it has practical applications across a wide variety of domains. In many modern real-world settings, graphs are extremely large, making it impractical to run algorithms that require access to the entire graph at once. This motivates the study of graph algorithms in the streaming model of computation; this is a commonly used model for handling large-scale or real-time data. There has been an extensive line of work on graph algorithms in the streaming setting. In addition to its practical utility, this line of work has led to a variety of new theoretical advances that have had auxiliary benefits well-beyond what was initially anticipated. 
Some well-studied problems in the streaming model include matching~\cite{McGregor05,GoelKK12,assadi2016maximum,assadi2017estimating,Kapralov21}, max-cut \cite{kapralov2014streaming}, spanners \cite{Baswana08,Elkin11, ahn2012graph,kapralov2014spanners,fernandez2020graph,filtser2021graph}, sparsifiers \cite{ahn2012graph,kapralov2014spanners}, shortest paths~\cite{FeigenbaumKMSZ08,guruswami2016superlinear}, and minimum spanning tree \cite{ahn2012graph,sun2015tight,nelson2019optimal}, among many others.

We consider the Survivable Network Design problem (SNDP). The input to this problem is an undirected graph $G = (V, E)$ with non-negative edge-weights $w: E \to \mathbb{R}_{\geq 0}$ and an integer connectivity requirement $r(uv)$ for each unordered pair of vertices $u, v \in V$. The objective is to find a minimum-weight subgraph $H \subseteq G$ such that, for every pair of vertices $u, v \in V$, there exist $r(uv)$ disjoint $uv$-paths in $H$. If the paths for the pairs are required to be {\em edge-disjoint}, the problem is referred to as {\em edge-connectivity} SNDP (EC-SNDP), and if they are required to be {\em vertex-disjoint}, it is known as the {\em vertex-connectivity} SNDP (VC-SNDP). We refer to the maximum connectivity requirement as $k := \max_{uv} r(uv)$. SNDP is a fundamental problem that generalizes many well-known polynomial-time solvable problems, including minimum spanning tree (MST) and $s$-$t$ shortest path, as well as several NP-hard problems, including Steiner Tree and Steiner Forest. We define some special cases of interest:
\begin{itemize}
  \item \emph{$k$-Connected Subgraph:} This is the special case in which $r(uv) = k$ for all $u,v \in V$. We denote the edge version as $k$-ECSS and the vertex version as $k$-VCSS; the goal is to find a min-weight $k$-edge-connected and $k$-vertex-connected spanning subgraph respectively. 
  \item \emph{Connectivity Augmentation:} In this problem, we are given a partial solution $G' = (V, E') \subseteq G$ for ``free'', with the guarantee that each $u, v \in V$ is at least $r(uv)-1$ connected in $G'$. The goal is to find a min-weight set of edges $F \subseteq E \setminus E'$ that \emph{augments} the connectivity of each $u,v$ pair to $r(uv)$. The edges $E \setminus E'$ are often referred to as \emph{links}. We refer to the augmentation version of the edge/vertex spanning problems as $k$-EC-CAP and $k$-VC-CAP respectively: here we are given a $k$-connected graph and the goal is to increase its connectivity to $k+1$. 
\end{itemize}
In this work, we study SNDP in the {\em insertion-only} streaming model. Formally, the algorithm reads the edges of a graph sequentially in an arbitrary order, processing each edge as it arrives. The goal is to solve graph problems over the streamed edges in a single pass, using a memory significantly smaller than storing the entire set of edges. 

While there has been significant recent progress for EC-SNDP in the streaming model, VC-SNDP remains essentially unexplored. One reason for this is that vertex-connectivity network design problems often do not share the same structural properties as their edge-connectivity counterparts. For instance, in the offline setting, EC-SNDP admits a $2$-approximation via a seminal iterated rounding algorithm of Jain \cite{jain2001factor}. However, these techniques  fail to extend to VC-SNDP. The best known approximation for VC-SNDP is $O(k^3 \log n)$ \cite{chuzhoy2009k}; moreover, the dependence on $k$ is known to be necessary \cite{chakraborty2008network}. Another example highlighting the difficulty of vertex-connectivity problems is $k$-CAP: $k$-EC-CAP is known to reduce to $1$ or $2$-EC-CAP for all $k$ \cite{DinitsKL76} (see also~\cite{KhullerT93,cheriyan19992}), but no such structure exists for the vertex-connectivity setting. Thus the main motivating question for this paper is the following:
\begin{question}\label{question:vc_sndp}
  What is the approximability of vertex-connectivity network design problems in the streaming setting?
\end{question}

We briefly discuss prior work on streaming algorithms for EC-SNDP and several of its special cases. As mentioned earlier, streaming algorithms for shortest path and MST are both well-studied. We note an important distinction between these two problems: while MST admits an exact algorithm in $O(n)$ words of space, the current best known streaming algorithm for $s$-$t$ shortest path is an $O(t)$ approximation in $\tilde{O}(n^{1+1/t})$-space. Thus, any problem that contains $s$-$t$ shortest path as a special case incurs this limitation,  while global connectivity problems such as MST are likely to have better trade-offs. Some other special cases of SNDP that have been studied in the streaming model include the Steiner forest problem in geometric setting~\cite{czumaj2022streaming}, and testing $k$-connectivity \cite{zelke2006k,zelke2011intractability,sun2015tight,assadi2021simple}. EC-SNDP in the streaming model was first studied in generality very recently in the work of \cite{JKMV24}. Their work primarily focused on $k$-EC-CAP in two models:
\begin{itemize}
  \item \emph{Link arrival:} In this model, the partial solution $G'$ is given up-front and does not count towards the space complexity of the algorithm. The weighted links $E \setminus E(G')$ arrive in the stream.
  \item \emph{Fully streaming:} In this model, the edges of the partial solution, along with the additional weighted links used in augmentation, both arrive in the stream. They may arrive in an interleaved fashion; that is, the algorithm may not know the full partial solution when processing some links in $E \setminus E(G')$. 
\end{itemize}
Both models are practically useful in their own right. Note that an $\alpha$-approximation for $k$-CAP in the link-arrival model implies an $(\alpha k)$-approximation for $k$-ECSS if we make $k$ passes over the stream, as we can augment the connectivity of the graph by one in each pass. This is particularly useful in situations where $k$ is a small constant and the best approximation ratio for link-arrival is significantly better than that of fully-streaming. 

In the link-arrival model, \cite{JKMV24} obtained a tight $(2+\eps)$-approximation for $k$-EC-CAP in $O(\frac{n\log n}{\eps})$ space; note that while one can obtain a better than $2$ in the offline setting, there is a lower bound of $2$ in the semi-streaming setting. In the fully streaming model, they obtained an $O(t)$-approximation in $\tilde O(kn + n^{1+1/t})$ space  for the connectivity augmentation problem
using a \emph{spanner} approach (a $t$-spanner is a sparse subgraph that preserves all distances to within a factor of $t$). By building on this, and using the reverse augmentation framework of Goemans {\em et al.}~\cite{GoemansGPSTW94}, they achieved an $O(t\log k)$-approximation for EC-SNDP with maximum connectivity requirement $k$; the space usage is $\tilde{O}(k n^{1+1/t})$. They further showed that any $O(t)$-approximation for this problem requires at least $\tilde{\Omega}(kn + n^{1+1/t})$ space.\footnote{While they only mentioned $\tilde{\Omega}(n^{1+1/t})$, it is straightforward to show that in general, the number of edges in an feasible solution of an SNDP instance with maximum connectivity requirement $k$ can be as large as $nk$. Hence, $nk$ is a trivial lower bound too.} Although the space complexity of their algorithm nearly matches the lower bound, their approximation for EC-SNDP is worse by a $\log k$ factor. A natural question is the following.

\begin{question}\label{question:ec_sndp}
   Is there an $O(t)$-approximation for EC-SNDP using $\tilde{O}(n^{1+1/t} + kn)$-space?
\end{question}


\subsection{Results and Techniques}
We make significant progress towards our motivating questions and obtain a number of results. We refer the reader to a summary of our results and prior work on streaming network design problems, provided in Tables~\ref{tab:results-ec} and~\ref{tab:results-vertex}. Note that all results in the table assume only one pass over the stream. In all weighted graphs $G = (V,E)$, we assume that the weight function $w: E \rightarrow \{0,1,\dots,W\}$ where $W = n^{\polylog(n)}$. This, in particular, implies that $\log W = \mathrm{poly}(\log n)$ and we will not mention $\log W$ explicitly in the space complexity. While all of our algorithms have a $\log W$ dependence in their space complexity (in terms of words of space), designing streaming algorithms where the total number of stored edges is independent of $W$, as explored in~\cite{JKMV24}, is an interesting technical question in its own right. 
We also note that all our algorithms assume an exact algorithm at the end of the stream, as the focus of this paper is optimizing the tradeoff between space and approximation ratio, rather than focusing on runtime. We refer the reader to Section \ref{sec:polytime} and Remark \ref{rem:vc_link_polytime} for a discussion on how these algorithms can be made efficient. 

\paragraph{A framework for SNDP:} Our first contribution is a general and broadly applicable framework 
to obtain streaming algorithms with low space and good approximation ratios for SNDP; this is provided in Section \ref{sec:generic}. This framework applies to both edge and vertex connectivity, as well as an intermediate setting known as element-connectivity (formally defined in Section \ref{sec:prelim}). This framework addresses Question \ref{question:vc_sndp} and partially resolves Question \ref{question:ec_sndp} affirmatively.\footnote{The upper and lower bounds match for several interesting/practical values of $k$, e.g. $k = O(\polylog n)$ or $k = \Omega(n)$. However, there is still a small gap between the bounds in general.}
All the results below use $\tilde{O}(k^{1-1/t}n^{1+1/t})$-space for vertex-connectivity and element-connectivity problems, and $\tilde{O}(k^{1/2-1/(2t)}n^{1+1/t} + kn)$-space for edge-connectivity problems, where $k$ is the maximum connectivity requirement and $t$ is a parameter that allows for an approximation vs. space tradeoff. \footnote{When $t$ is an even integer, the space used is a factor of $k^{1/(2t)}$ more than when $t$ is odd. The results stated in this section assume that $t$ is an odd integer.}

\begin{itemize}
\item For EC-SNDP, the framework yields $8t$-approximation, improving upon the $O(t \log k)$-approximation from \cite{JKMV24}. For an $O(t)$-approximation, the space bound of our algorithm is off by only a factor of $k^{1/2 - 1/(2t)}$ from the known lower bound.
These bounds also hold for element-connectivity, providing the first such results for this variant. 
\item For VC-SNDP, the framework yields an $O(tk)$-approximation.
\item For VC-SNDP, the framework yields an $O(\beta t)$-approximation where $\beta$ is the integrality gap of the natural LP relaxation. 
Using this, we obtain improved algorithms for several important special cases including VC-SNDP when $k \le 2$, $k$-VCSS, and $k$-VC-CAP.  
\end{itemize}
\begin{remark}
  Extending the lower bound construction from~\cite{JKMV24} to the vertex connectivity setting, we show that approximating VC-SNDP within a factor better than $2t+1$ requires $\Omega(nk + n^{1+1/t})$ space, which points out that  our upper bounds are nearly tight. The formal statement of the lower bound is in Section~\ref{sec:vc-sndp-lb}. 
\end{remark}
\begin{remark}
  This framework also extends to \emph{non-uniform} network design models. By Menger's theorem (see Section \ref{sec:prelim}), the goal in SNDP is to construct graphs that maintain connectivity between terminal pairs despite the failure of \emph{any} $k-1$ edges/vertices. Non-uniform network design models scenarios in which only certain specified subsets of edges can fail. This model is relevant in settings where edges failures are correlated in some way. For brevity we omit a detailed discussion here, and instead refer the reader to the following works on some problems non-uniform network design for which our framework holds: Bulk-Robust Network Design \cite{AdjiashviliSZ15,Adjiashvili15}, Flexible Graph Connectivity \cite{Adjiashvili13,AdjiashviliHM22,AdjiashviliHMS20,BoydCHI22,BansalCGI22,bansal2023constant,ChekuriJ23,nutov2024_fgc,HJS24_fgc}, and Relative Survivable Network Design \cite{DinitzKK22,DinitzKKN23}. We note that while these problems have mostly been studied in edge-connectivity settings, our streaming framework can also handle vertex-connectivity versions.
\end{remark}

\paragraph{Algorithms for $k$-VC-CAP:}
Our second contribution is for $k$-VC-CAP in the link-arrival model; this is provided in Section \ref{sec:vc_link_arrival}. Note that this is a global connectivity problem
and does not include the $s$-$t$-shortest path problem as a special case; hence we hope to obtain better bounds
that avoid the overhead of using spanners. 
We obtain the following approximation ratios for $k$-VC-CAP; both algorithms use $\tilde O(n/\eps)$ space.
\begin{itemize}
    \item For $k = 1$, we obtain a $(3 + \eps)$-approximation algorithm. 
    \item For $k = 2$, we obtain a $(7 + \eps)$-approximation algorithm.
\end{itemize}
Following the earlier discussion, this implies that with $k$ passes of the stream, we obtain a constant-factor approximation in near-linear space for $k$-VCSS for $k \leq 3$.

\begin{table}[!ht]
\centering
{\renewcommand{\arraystretch}{1.5}%
\begin{tabular}{c|c|c|c}\toprule
{\bf Problem} & {\bf Approx.} & {\bf Space} & {\bf Note} \\ \midrule\midrule
\multirow{4}{*}{$k$-EC-CAP} & $2+\eps$ & $O(\frac{n}{\eps}\log n)$~\cite{JKMV24} & {link arrival} \\
& $2-\eps$ & $\Omega(n^2)$ bits~\cite{JKMV24} & {link arrival lower bound}\\
\cmidrule{2-4}
& \multirow{2}{*}{$O(t)$} & $\tilde{O}(n^{1+\frac{1}{t}} + kn)$~\cite{JKMV24}  & {fully streaming}\\
& & $\Omega(n^{1+\frac{1}{t}} + kn)$ bits~\cite{JKMV24} & {fully streaming lower bound}\\ 
\midrule\midrule

\multirow{3}{*}{EC-SNDP} & $O(t \log k)$ & $\tilde{O}(kn^{1+\frac{1}{t}})$~\cite{JKMV24} &  \\ 
& $8t$ & $\tilde{O}(k^{\frac{1}{2}-\frac{1}{2t}} n^{1+\frac{1}{t}} + kn)$ \textcolor{blue}{[Here]} &  \\ 
& $O(t)$ & $\Omega(n^{1+\frac{1}{t}} + kn)$ bits~\cite{JKMV24} & lower bound\\ 
\midrule
\multirow{2}{*}{$k$-ECSS} & $2$ & $O(kn)$~\cite{zelke2006k,cheriyan1993scan,nagamochi1992linear} & for unweighted graphs \\ 
\cmidrule{2-4}
&  $8t$ &$\tilde{O}(k^{\frac{1}{2}-\frac{1}{2t}} n^{1+\frac{1}{t}} + kn)$ \textcolor{blue}{[Here]} &  \\
& $O(t)$ & $\Omega(n^{1+\frac{1}{t}} + kn)$ bits~\cite{JKMV24} & lower bound\\ 
\bottomrule
\end{tabular}
}
\caption{Summary of results for edge-connectivity network design.}
\label{tab:results-ec}
\end{table}

\begin{table}[!ht]
\centering
{\renewcommand{\arraystretch}{1.5}%
\begin{tabular}{c|c|c|c}\toprule
{\bf Problem} & {\bf Approx.} & {\bf Space} & {\bf Note} \\ \midrule\midrule
\multirow{4}{*}{$k$-VC-CAP} & $O(t)$ & $\tilde{O}(k^{1 - \frac{1}{t}} n^{1+\frac{1}{t}})$ \textcolor{blue}{[Here]} & fully streaming ($n = \Omega(k^3)$) 
\\
& $O(t)$ & $\Omega(n^{1+\frac{1}{t}})$ bits \textcolor{blue}{[Here]}& fully streaming lower bound \\ 
& $3 + \eps$ & $\tilde O(n/\eps)$ \textcolor{blue}{[Here]} & $k = 1$ (link arrival) \\
& $7 + \eps$ & $\tilde O(n/\eps)$ \textcolor{blue}{[Here]} & $k = 2$ (link arrival) \\
\midrule\midrule
\multirow{3}{*}{VC-SNDP} & $2tk$ & $\tilde{O}(k^{1 - \frac{1}{t}} n^{1+\frac{1}{t}})$ \textcolor{blue}{[Here]} &  \\ 
& $8t$ & $\tilde{O}(n^{1+\frac{1}{t}})$ \textcolor{blue}{[Here]}& $\{0,1,2\}$-VC-SNDP \\
& $O(t)$ & $\Omega(n^{1+\frac{1}{t}} + kn)$ bits \textcolor{blue}{[Here]} & lower bound\\ 
\midrule
\multirow{4}{*}{$k$-VCSS} & $2$ & $O(nk)$~\cite{zelke2006k,cheriyan1993scan,nagamochi1992linear} & for unweighted graphs \\
\cmidrule{2-4}
& $O(t)$ & $\tilde{O}(k^{1-\frac{1}{t}} n^{1+\frac{1}{t}})$ \textcolor{blue}{[Here]} & $n = \Omega(k^3)$
\\ 
& $O(t)$ & $\Omega(n^{1+\frac{1}{t}} + kn)$ bits \textcolor{blue}{[Here]} & lower bound \\ 
\bottomrule
\end{tabular}
}
\caption{Summary of our results for vertex-connectivity network design.}
\label{tab:results-vertex}
\end{table}

\paragraph{Techniques:} 
Our general framework for SNDP and related problems is based on a connection to \emph{fault-tolerant spanners} which were first introduced in~\cite{levcopoulos1998efficient} in geometric settings and then extensively studied in the graph setting. These objects generalize the notion of spanners to allow faults and are surprisingly powerful. Recent work has shown that, similar to ordinary spanners, optimal vertex fault-tolerant spanners and near-optimal edge fault-tolerant spanners can be constructed using a simple greedy algorithm, which enables their use in the streaming setting~\cite{bodwin2019trivial,bodwin2018optimal}.
Using these spanners alone suffices to obtain an $O(tk)$-approximation assuming an exact algorithm at the end of the stream. Obtaining our refined approximation bounds requires additional insight: we combine the use of fault-tolerant spanners with an analysis via natural LP relaxations for SNDP. This improved analysis has a key benefit. It allows us to directly improve the factor $O(tk)$ for problems with corresponding integrality gap better than $k$---this applies to several of the problems we consider, such as EC-SNDP, ELC-SNDP, $k$-VCSS and $\{0,1,2\}$-VC-SNDP. 

For $1$-VC-CAP, we aim to augment a spanning tree to a $2$-vertex-connected graph. To do so, we borrow ideas from the edge-connectivity setting: we root the tree arbitrarily, and for each node $u \in V$, we store the link that ``covers'' the most edges in the path from $u$ to the root. This does not quite suffice in the vertex-connectivity setting, thus some additional care is required. Our main contribution is an algorithm for $2$-VC-CAP. Here we need to augment a $2$-vertex-connected graph to a $3$-vertex-connected graph. In edge-connectivity, this reduces to cactus-augmentation (which can essentially be reduced to augmenting a cycle). However this does not hold for the vertex-connectivity setting. Instead, we use the SPQR-representation of a $2$-vertex-connected graph: this is a compact tree-like data structure that captures all the 2-node cuts of a graph and was introduced by Di Battista and Tamassia \cite{spqr-tree} that dynamically maintains the triconnected components (and thus all $2$-node cuts) of a graph. It was initially developed in the context of planar graphs~\cite{di1989incremental,di1996line}. We use this tree-like structure to combine ideas from $1$-VC-CAP with ideas from cactus augmentation in the edge-connectivity setting to obtain our result. The algorithm is tailored to the particulars of the SPQR representation and is technically quite involved. We refer the reader to
Figure~\ref{fig:vc_2_to_3_SPQR} for an example of a $2$-connected graph and its SPQR tree.

\subsection{Related Work}
\label{sec:rel-work}

\paragraph{Offline Network Design:} 
There is substantial literature on network design; here we briefly discuss some relevant literature.
EC-SNDP admits a $2$-approximation via the seminal work on iterated rounding by Jain \cite{jain2001factor} and this has been extended to ELC-SNDP \cite{fleischer2006iterative,CheriyanVV06}. We note that the best known approximation for Steiner Forest, the special case with connectivity requirements
in $\{0,1\}$, is $2$. Steiner tree is another important special case, and for this there is a $(\ln 4 +\eps)$-approximation \cite{byrka2013steiner}. 
Even for $2$-ECSS the best known approximation is $2$, although better bounds are known for the unweighted case.
Recently there has been exciting progress on $k$-EC-CAP, starting with progress on the special case of weighted TAP (tree augmentation problem which 
corresponds to $k=1$). For all $k$, $k$-EC-CAP admits a ($1.5 +\eps)$-approximation \cite{TraubZ23}---see also~\cite{traub2022better,traub2022local,byrka2020breaching,cecchetto2021bridging,garg2023improved}. 
In contrast to the constant factor approximation results for edge-connectivity, the best known approximation for VC-SNDP is $O(k^3 \log n)$ due to Chuzhoy and Khanna \cite{chuzhoy2009k}. Moreover, even for the single-source setting it is known that the approximation ratio needs to depend on $k$ for sufficiently large $k$~\cite{chakraborty2008network}. The single-source problem admits an $O(k^2)$-approximation \cite{Nutov12,nutov-erratum} and subset $k$-connectivity admits
an $O(k\log^2 k)$-approximation \cite{laekhanukit2015improved}.
Several special cases have better approximation bounds. When $k \le 2$, VC-SNDP admits a $2$-approximation \cite{fleischer2006iterative,CheriyanVV06} which also implies
that $k$-VCSS for $k=2$ and $k$-VC-CAP for $k=1$ admit a  $2$-approximation.
For $k$-VCSS a $(4+\eps)$-approximation is known
when $n$ is large compared to $k$ \cite{Nutov22,CheriyanV14}. For smaller value of $k \le 6$
improved bounds are known---see \cite{nutov2018improved}. These are also the best known bounds for $k$-VC-CAP when $n$ is large compared to $k$.
For large $k$, $k$-VC-CAP and $k$-VCSS admit an $O(\log \frac{n}{n-k})$ and $O(\log k \log \frac{n}{n-k})$-approximation respectively with more precise bounds known in various regimes~\cite{nutov2018improved}.
We note that many of the approximation results (but not all) are with respect to a natural cut-cover LP relaxation. This is important to
our analysis since some of our results are based in exploiting the integrality gap of this LP relaxation. Many improved results are known
for special cases of graphs including unweighted graphs, planar and graphs from proper minor-closed families, and graphs arising from geometric instances. We focus on general graphs in this work.

\paragraph{Streaming Graph Algorithms:} 
Graph problems in the streaming model have been studied extensively, particularly in the context of compression methods that reduce the graph size and preserve connectivity within a factor of $t$ (i.e., {\em spanners})~\cite{feigenbaum2005graph,Baswana08,Elkin11} or approximate cuts within a $(1+\eps)$ factor (i.e., {\em cut sparsifiers} and related problems)~\cite{ahn2009graph,kelner2013spectral,KapralovLMMS17,KapralovMMMNST20}. These problems have also been widely explored in the {\em dynamic} streaming model, where edges are both inserted and deleted. While graph sketching approaches have sufficed to yield near-optimal algorithms for sparsifiers~\cite{ahn2012graph}, the state-of-the-art for spanners in dynamic streams remained multi-pass algorithms until recently~\cite{ahn2012graph,kapralov2014streaming}. Filtser {\em et al.}~\cite{filtser2021graph} developed the first single-pass algorithm for $\tilde{O}(n^{2/3})$-spanners using $\tilde{O}(n)$ space in dynamic streams. Though this result has a large approximation factor, Filtser {\em et al.} conjectured that it may represent an optimal trade-off.    

Another line of research related to our work is testing connectivity in graph algorithms. This problem has been studied for both edge-connectivity and vertex-connectivity~\cite{feigenbaum2005graph,zelke2006k,sun2015tight}, as well as in dynamic settings~\cite{ahn2012graph,CrouchMS13,guha2015vertex,assadi2023tight}. Specifically, $k$-connectivity for both edge- and vertex-connectivity can be tested using $\tilde{O}(nk)$ space even in dynamic streams~\cite{ahn2012graph,assadi2023tight}. 

Finally, for the problem of $k$-ECSS,~\cite{JKMV24} designed a $k$-pass algorithm within the augmentation framework~\cite{GoemansGPSTW94} that finds an $O(\log k)$-approximate solution using $\tilde{O}(nk)$ space.

\section{Preliminaries}
\label{sec:prelim}
In this section, we provide preliminary background on connectivity. 
In a graph $G$, two vertices $s$ and $t$ are $k$-edge ($k$-vertex) connected if $G$ contains $k$ edge-disjoint (vertex-disjoint) $st$-paths.
There is a close relation between the maximum connectivity of a pair of vertices $s, t$, and the minimum cut separating $s$ and $t$, as characterized by Menger's theorem. 
Let $S\subset V(G)$ be a subset of vertices in $G$. We denote the set of edges crossing $S$ by $\delta_G(S)$; $\delta_G(S)=\{uv\in E(G)\;|\; u\in S, v\in V\setminus S\}$. We drop $G$ when it is implicit from the context.
Menger's theorem is a key result in problems involving connectivity requirements. Its formulation in terms of edge connectivity is as follows:
\begin{theorem}[Edge-connectivity Menger's theorem]\label{thm:ec-menger}
Let $G=(V, E)$ be an undirected graph. Two vertices $s,t \in V$ are $k$-edge connected iff for each set $S\subset V$ such that $s\in S$ and $t\in V\setminus S$, $|{\delta(S)}|\geq k$.
\end{theorem}

\noindent
To state Menger's theorem for vertex-connectivity requirements, we first introduce some notation known as a {\em biset}, following prior work on vertex connectivity~\cite{Nutov12,CheriyanVV06,fleischer2006iterative}. A biset $\hat{X} = (X, X^+)$ is a pair of sets where $X \subseteq X^+ \subseteq V$. 
We say that an edge crosses a biset $\hat{S}$ if one of its endpoints is in $S$ and the other is in $V \setminus S^+$. We define $\delta(\hat{S}) = \{uv \in E(G) \mid u \in S, v \in V \setminus S^+\}$.

\begin{theorem}[Vertex-connectivity Menger's theorem]\label{thm:vc-menger}
Let $G=(V,E)$ be an undirected graph. Two vertices $s, t \in V$ are $k$-vertex connected iff for each biset $\hat{S}\subset V\times V$ such that $s\in S$ and $t\in V\setminus S^+$, $|\delta(\hat{S})| + |S^+\setminus S|\geq k$.
\end{theorem}

An intermediate connectivity notion between edge-connectivity and vertex-connectivity, proposed by Jain {\em et al.}~\cite{jain2002primal}, is {\em element-connectivity}, in which the input set of vertices $V$ is divided into two types: {\em reliable} ($R$) and {\em non-reliable} ($V\setminus R$). Elements denote the set of edges $E$ and the set of non-reliable vertices $V\setminus R$.
For a pair of vertices in $u,v\in R$, a set of $uv$-paths are \emph{element-disjoint} if they are disjoint on elements (i.e., edges and non-reliable vertices). Note that the main distinction between element and vertex connectivity is that element-disjoint paths are not necessarily disjoint in the reliable vertices, and the requirements are only on reliable vertices.

Menger's theorem for element-connectivity can be stated as follows:
\begin{theorem}[Element-connectivity Menger's theorem]\label{thm:elc-menger}
Let $G=(V,E)$ be an undirected graph, whose vertices are partitioned into reliable ($R$) and non-reliable ($V\setminus R$). Two vertices $s, t \in R$ are $k$-element connected iff for each biset $\hat{S}\subset V\times V$ such that $s\in S$, $t\in V\setminus S^+$, and $S^+\setminus S$ containing only {\em non-reliable} vertices (i.e., $S^+\setminus S \subseteq V\setminus R$), we have $|\delta(\hat{S})| + |S^+\setminus S|\geq k$.
\end{theorem}

\subsection{Fault-Tolerant Spanners in Streaming}\label{sec:FT-spanner-streaming}

\begin{definition}[Fault-Tolerant Spanners]
A subgraph $H \subseteq G$ is an {\em $f$-vertex-fault-tolerant ($f$-VFT) $t$-spanner} of $G$ if, for every subset of vertices $F_V \subseteq V$ of size at most $f$, all pairwise distances in $V\setminus F_V$ are preserved within a factor of $t$. That is, for all $u, v \in V\setminus F_V$,
\begin{align*}
    d_{H\setminus F_V}(v, u) \le t \cdot d_{G\setminus F_V}(v, u),
\end{align*}
where $G \setminus F_V$ and $H \setminus F_V$ denote the induced subgraphs $G[V \setminus F_V]$ and $H[V \setminus F_V]$, respectively.

Similarly, a subgraph $H \subseteq G$ is an {\em $f$-edge-fault-tolerant ($f$-EFT) $t$-spanner} of $G$ if, for every subset of edges $F_E \subseteq E$ of size at most $f$, all pairwise distances in $V$ are preserved within a factor of $t$. That is, for all $u, v \in V$,
\begin{align*}
    d_{H\setminus F_E}(v, u) \le t \cdot d_{G\setminus F_E}(v, u),
\end{align*}
where $G \setminus F_E$ and $H \setminus F_E$ denote the induced subgraphs $G[E \setminus F_E]$ and $H[E \setminus F_E]$, respectively.
\end{definition}

A natural algorithm for constructing fault-tolerant spanners, which is a straightforward adaptation of the standard greedy algorithm of~\cite{althofer1993sparse} for spanners, is to process the edges in an arbitrary order and add an edge $(u,v)$ in the spanner if there exists a set of vertices of size at most $f$ such that their removal increases the distance of $u,v$ in the so-far-constructed spanner to at least $t+1$ (refer to Algorithm~\ref{alg:FT-spanner}). 

\begin{algorithm}
    \KwIn{Unweighted graph $G = (V, E)$, a fault-tolerance parameter $f$, and a stretch parameter $t$.}
    \caption{The greedy algorithm for unweighted VFT (EFT) spanners.}
    \label{alg:FT-spanner}
    {\bf initialize} $H \leftarrow (V, \emptyset)$
    
    \For{$(u,v) \in E$ in an arbitrary order}{
        \If{there exists a set of vertices (or edges) $F$ of size $f$ such that $d_{H\setminus F}(u,v) > t$}{
            {\bf add} $(u,v)$ to $H$    
        }
    }
    \Return $H$
\end{algorithm}

\begin{theorem}[\cite{bodwin2019trivial}]\label{thm:vft-spanner} 
    For any $n$-vertex graph $G$, the greedy algorithm for the $f$-VFT $(2t-1)$-spanner (i.e., Algorithm~\ref{alg:FT-spanner}) returns a feasible subgraph $H$ of size $|E(H)| = O(f^{1-1/t} \cdot n^{1 + 1/t})$. 
\end{theorem}
This bound on spanner size is tight and fully matches the known lower bound for VFT~\cite{bodwin2018optimal}. 
\begin{theorem}[\cite{bodwin2022partially}]\label{thm:eft-spanner} 
    For any $n$-vertex graph $G$, the greedy algorithm for the $f$-EFT $(2t-1)$-spanner (i.e., Algorithm~\ref{alg:FT-spanner}) returns a feasible subgraph $H$ of size 
    \begin{align*}
    |E(H)| =
    \begin{cases}
        O\left(t^2 f^{1/2 - 1/(2t)} n^{1 + 1/t} + tfn \right) & \text{$t$ is odd} \\
        O\left(t^2 f^{1/2} n^{1 + 1/t} + tfn \right) & \text{$t$ is even}.
        \end{cases}
    \end{align*}
\end{theorem}

The bound on the spanner size is tight for constant odd $t$,
and is off from the best known lower bound~\cite{bodwin2018optimal} by only quadratic factors of $t$ for non-constant odd
$t$, and only $t^2 f^{1/(2t)}$ for even $t$.

It is straightforward to show that for unweighted graphs, the greedy algorithm for spanners can be implemented in insertion-only streams with a space complexity equal to the size of $H$. Moreover, for the remainder of the paper, we assume that $t = O(\log n)$, as setting it any larger does not yield asymptotic improvements in the space complexity of the spanner.

\paragraph{Weighted graphs.} Although running the greedy algorithm 
with edges in the increasing order of edge weights provides the same guarantee on size and stretch for fault-tolerant spanners (cf.~\cite{bodwin2019trivial}), the algorithm is no longer implementable in the streaming setting. 
A standard technique to address this issue is {\em bucketing}. Given a weighted graph $G = (V, E)$ with weight function $w: E \rightarrow \{0, 1, \cdots, W\}$, we partition the edges of $G$ into $O(\epsilon^{-1} \log W)$ buckets, where the $i$th bucket contains edges with weights in the range $B_ i\coloneqq\left[(1+\epsilon)^{i-1}, (1+\epsilon)^i\right)$, for $i\ge 1$. Furthermore, we use bucket $B_{0}$ for zero-weight edges. 
Then, we construct a fault-tolerant spanner in each bucket (treating it as an unweighted graph) using the greedy algorithm for unweighted $f$-FT spanners. 

\begin{algorithm}
    \KwIn{Weighted graph $(G = (V, E), w)$, a stretch parameter $t$, and a fault tolerance parameter $f$.}
    \caption{The streaming greedy algorithm for VFT (or EFT) spanners on weighted graphs.}
    \label{alg:FT-spanner-weighted}
    \For{$i=1$ to $T = O(\eps^{-1} \log W)$}{
        {\bf initialize} $H_i \leftarrow (V, \emptyset)$
    }
       
    \For{$(u,v) \in E$ in the stream}{
        \tcc{let $H_j$ be the FT spanner corresponding to the weight class of $(u,v)$; i.e. $w(u,v) \in B_j$}
        
        \If{there exists a set of vertices (or edges) $F$ of size $f$ such that $d_{H_j\setminus F}(u,v) > 2t-1$}{
            {\bf add} $(u,v)$ to $H_j$    
        }
    }
    \Return $H \coloneqq H_1 \cup \cdots \cup H_T$
\end{algorithm}

\begin{theorem}[Streaming Weighted VFT Spanners]\label{thm:weighted-VFT-spanner-stream} 
    There exists a single-pass streaming algorithm (i.e., Algorithm~\ref{alg:FT-spanner-weighted}) that uses $O(f^{1-1/t} \cdot n^{1+1/t} \cdot \eps^{-1} \cdot \log W)$ words of space and returns an $f$-VFT  $\big((1+\eps)(2t-1)\big)$-spanner of an $n$-vertex weighted graph $G=(V,E)$ with the weight function $w:E\rightarrow \{0, 1, \cdots, W\}$, of size $O(f^{1-1/t} \cdot n^{1+1/t}\cdot \eps^{-1} \cdot \log W)$. 

    In particular, if $W = n^{O(\log n)}$ by setting $\eps = 1/(2t-1)$, the algorithm uses $\tilde{O}(f^{1-1/t} \cdot n^{1+1/t})$ words of space and returns an $f$-VFT $(2t)$-spanner of size $\tilde{O}(f^{1-1/t} \cdot n^{1+1/t})$. 
\end{theorem}

\begin{theorem}[Streaming Weighted EFT Spanners]\label{thm:weighted-EFT-spanner-stream} 
    There exists a single-pass streaming algorithm (i.e., Algorithm~\ref{alg:FT-spanner-weighted}) that uses $O\big( (f^{1/2-1/(2t)} \cdot n^{1+1/t} + fn) \cdot \eps^{-1} \cdot \log W\big)$ words of space and returns an $f$-EFT $\big((1+\eps)(2t-1)\big)$-spanner of an $n$-vertex weighted graph $G=(V,E)$ with the weight function $w:E\rightarrow \{0, 1, \cdots, W\}$, of size $O\big( (f^{1/2-1/(2t)} \cdot n^{1+1/t} + fn) \cdot \eps^{-1} \cdot \log W\big)$, when $t$ is an odd integer. 

    In particular, if $W = n^{O(\log n)}$ by setting $\eps = 1/(2t-1)$, the algorithm uses $\tilde{O}(f^{1/2-1/(2t)} \cdot n^{1+1/t} + fn)$ words of space and returns an $f$-EFT $(2t)$-spanner of size $\tilde{O}(f^{1/2-1/(2t)} \cdot n^{1+1/t} + fn)$. 
\end{theorem}
\section{Generic Framework for Streaming Algorithms for Network Design}\label{sec:generic}
In this section, we describe our generic framework for streaming algorithms for network design problems. The algorithm, described in Algorithm~\ref{alg:generic-stream-network-design}, is both simple and practical.
\begin{algorithm}
    \KwIn{Stream of weighted edges $\big(e, w(e)\big)$, network design problem $\mathcal{M}$ with maximum connectivity requirement $k$, approximation parameter $t$, and ``offline'' algorithm $\mathcal{A}$ for $\mathcal{M}$.}
    \caption{A Generic framework for solving network design problems in streaming setting.}
    \label{alg:generic-stream-network-design}
    
    {\large \textbf{\tcc{\sc In Stream:}}}
    {\bf construct} an $O(kt)$-FT $(2t-1)$-spanner $H$ of the edges using Algorithm~\ref{alg:FT-spanner-weighted} with $\eps = \frac{1}{2t-1}$.

   
    {\large \textbf{\tcc{{\sc Postprocessing} {\normalsize (\emph{after stream terminates})}:}}}
    \Return the solution $\sol$ of $\mathcal{M}$ on $H$ output by the algorithm $\mathcal{A}$  
\end{algorithm}

As the algorithm only constructs and stores a fault-tolerant spanner with the given parameters $f = O(tk)$ and $t$ in the stream, its space complexity follows directly from Theorems~\ref{thm:weighted-VFT-spanner-stream} and~\ref{thm:weighted-EFT-spanner-stream}. Specifically,
\begin{itemize}
    \item In VC-SNDP, the space complexity is $\tilde{O}\big(f^{1-1/t}\cdot n^{1+1/t}\big) = \tilde{O}\big(k^{1-1/t}\cdot n^{1+1/t}\big)$.
    \item In EC-SNDP, the space complexity is $\tilde{O}\big(f^{1/2-1/(2t)}\cdot n^{1+1/t} + fn\big) = \tilde{O}\big(k^{1/2-1/(2t)}\cdot n^{1+1/t} + kn\big)$.
\end{itemize}
However, the approximation \emph{analysis} of the algorithms is technical. 

The rest of this section is dedicated to analyzing the approximation performance of Algorithm~\ref{alg:generic-stream-network-design}, and we will mainly focus on vertex-connectivity requirements. For the majority of this section, we assume the offline algorithm $\mathcal A$ used at the end of the stream is an exhaustive search: we enumerate all subgraphs of $H$ and check feasibility for each, returning the solution of minimum cost. Note that this is possible since feasibility of SNDP (for both edge and vertex connectivity) can be checked in linear space by checking connectivity via augmenting path based max-flow routines. In Section \ref{sec:polytime} we discuss some alternate choices of algorithm $\mathcal A$ and the resulting tradeoffs of runtime, space, and approximation ratio. 

Before analyzing our framework, we present an observation on the structure of VFT spanners (or EFT spanners) 
that is used in our analysis for all variants. 

\begin{observation}\label{obs:vft-spanner-structure}
    Given a weighted graph $G = (V,E)$, let $H$ denote the VFT spanner of $G$ constructed by Algorithm~\ref{alg:FT-spanner-weighted} with parameters $(t, f, \eps)$. 
    If $e \in E$ with $w(e) \in B_j$ does not belong to $H_j$, i.e., $e \notin H_j$, then $H_j$ contains at least $L = \lfloor f/(t-1) \rfloor + 1$ vertex-disjoint $uv$-paths, each containing at most $t$ edges that all have weights in $B_j$. 
\end{observation}
\begin{proof}
To find the vertex-disjoint paths $P_1, \dots, P_L$, we work with the fault-tolerant spanner $H_j$ corresponding to $B_j$. We perform $L$ iterations, and in each iteration $i$, we find a $uv$-path $P_i$ of length at most $t$ that is vertex-disjoint from the previously constructed $uv$-paths $P_1, \dots, P_{i-1}$. To do this, we define the set of vertices $F_i = \big(P_1 \cup \dots \cup P_{i-1}\big)\setminus \{u,v\}$ and find a $uv$-path in $H_j \setminus F_i$. Initially, set $F_1 = \emptyset$. Since $|F_i| \le (t-1)(i-1) \le (t-1)(L-1) \le f$, by the properties of the $f$-VFT $t$-spanner $H_j$, there exists a $uv$-path $P_i \subseteq H_j \setminus F_i$, which is a path containing at most $t$ edges, each with weight belonging to $B_j$, and is vertex-disjoint from $P_1, \dots, P_{i-1}$. 
\end{proof}

\smallskip

\subsection{Vertex Connectivity Network Design}\label{sec:vc-sndp}
In this section, we consider the VC-SNDP in insertion-only streams. First, in Section~\ref{sec:vc-simple-analysis}, we present a simple analysis of our generic FT spanner based algorithm, Algorithm \ref{alg:generic-stream-network-design}, which yields a $kt$-approximation, where $k$ is the maximum connectivity requirement in the SNDP instance and $t$ is the stretch parameter in the VFT-spanner. 
Then, in Section~\ref{sec:vc-sndp-improved}, we provide a more involved analysis that achieves an approximation based on the integrality gap of the well-studied
cut-based LP relaxation for VC-SNDP. This yields improved results in several special cases of VC-SNDP, such as $k$-VCSS, $k$-VC-CAP and $\{0,1,2\}$-VC-SNDP, with small integrality gap.
We also show that this approach yields near-tight bounds for EC-SNDP and ELC-SNDP due to the small integrality gaps these problems have. 

\subsubsection{A Simple Analysis Based on Integral Solutions}\label{sec:vc-simple-analysis}

\begin{theorem}\label{thm:FT-spanner-VC-SNDP}
    Let $H$ be the VFT spanner of a weighted graph $G$ as constructed in Algorithm~\ref{alg:FT-spanner-weighted} with parameters $(t, f = (2t-2)(k-1), \eps = 1/(2t-1))$. 
    Then an optimal solution of VC-SNDP on ($H, r$) is within a $\big(2tk\big)$-factor of an optimal solution of VC-SNDP on ($G, r$).
\end{theorem}

\begin{proof}
    Let $H = H_1 \uplus \cdots \uplus H_T$, where each $H_j$ is a $\big((2t-2)(k-1)\big)$-VFT $(2t-1)$-spanner for the edges in $G$ with weights in the range $B_j \coloneqq ((1+\epsilon)^{j-1}, (1+\epsilon)^j]$, and $T = \epsilon^{-1} \log W = \mathrm{poly}(\log n)$.
    Let $\opt \subseteq E$ be an optimal solution for the VC-SNDP instance on $G$. We construct a feasible solution $\sol \subseteq H$ based on $\opt$ as follows.
    For every $e = (u, v) \in \opt$, if $e \in H_j$, add $e$ to $\sol$. Otherwise, add $k$ vertex-disjoint $uv$-paths in $H_j$ (as shown in Observation \ref{obs:vft-spanner-structure}) to $\sol$. Note that our algorithm does not need to explicitly construct $\sol$; this is only for analysis purposes.
    
    We prove the feasibility of $\sol$ for the VC-SNDP instance via Menger's theorem (Theorem~\ref{thm:vc-menger}). Consider a biset $\hat{S} = (S, S^+)$ with connectivity requirement $k' = \max_{v \in S, u \in V \setminus S^+} r(u, v) \leq k$, where $|S^+ \setminus S| = \gamma < k'$. The optimal solution, $\opt$, will have at least $k'-\gamma$ edges crossing this biset, i.e., $k'-\gamma$ edges with one endpoint in $S$ and another one in $V\setminus S^+$. Let these edges be denoted as $e_1, \dots, e_{k'-\gamma}$. If all these edges exist in $H$, then since $\sol$ also includes them, it trivially satisfies the connectivity requirement of $\hat{S}$ as well. Suppose, without loss of generality, that $e_1 = (u,v) \notin H$. Then, by Theorem~\ref{obs:vft-spanner-structure}, $\sol$  contains $k$ vertex-disjoint $uv$-paths from the same weight class as $e$; thus at least $k - \gamma \ge k'-\gamma$ of them cross $\hat{S}$. 

    For the cost analysis, note that for every edge $e = (u, v) \in \opt$ in weight class $B_j$, either $e \in \sol$, or $\sol$ adds at most $k$ vertex-disjoint $uv$-paths, $P_1, \dots, P_k$, each of length at most $2t-1$, with weights belonging to $B_j$. Therefore, $w(\sol) \le k \cdot (2t-1) \cdot (1 + \frac{1}{2t-1}) \cdot w(\opt) = 2tk \cdot w(\opt)$.
\end{proof}

\begin{corollary}\label{cor:stream-vc-sndp}
    There exists an algorithm for VC-SNDP with edge weights $w: E \rightarrow \{0, 1, \dots, W\}$ and a maximum connectivity requirement $k$, in insertion-only streams, that uses $\tilde{O}(k^{1-1/t}\cdot n^{1+1/t})$ space and outputs a $(2tk)$-approximate solution. 
\end{corollary}
\begin{proof}
Theorem~\ref{thm:FT-spanner-VC-SNDP} shows that Algorithm~\ref{alg:generic-stream-network-design} with $(t, f = (2t-2)(k-1), \eps = 1/(2t-1))$, 
returns a VFT spanner $H$ 
of size $\tilde{O}(k^{1-1/t} \cdot n^{1+1/t})$ that contains a $(2tk)$-approximate solution for VC-SNDP on $G$. 
Then, once the stream terminates, we perform an exhaustive search, i.e., enumerating all possible solutions on $H$ and output a $(2tk)$-approximate solution of $G$, supported on $H$ only. 
Since the algorithm's only space consumption is for storing the constructed VFT spanner, it requires $\tilde{O}\big(k^{1-1/t} \cdot n^{1+1/t}\big)$ space.
\end{proof}

\subsubsection{An Improved Analysis via Fractional Solutions}\label{sec:vc-sndp-improved} 
In this section we show a refined analysis of the framework which shows that an $O(tk)$-VFT $O(t)$-spanner contains a $O(t)$-approximate \emph{fractional} 
solution for the given VC-SNDP instance. This is particularly interesting as it demonstrates that a fault-tolerant spanner preserves a near-optimal fractional solution 
for VC-SNDP on the given graph $G$. In other words, fault-tolerant spanners serve as {\em coresets} for both integral and fractional solutions of network design problems.

\begin{theorem}\label{thm:FT-spanner-VC-SNDP-frational}
    Let $H$ be the VFT spanner of a weighted graph $G$ as constructed in Algorithm~\ref{alg:FT-spanner-weighted} with parameters $(t, f = (2t-2)(2k-1), \eps = 1/(2t-1))$. Then the weight of an optimal \emph{fractional} solution of VC-SNDP on ($H, r$) is within a $4t$-factor of the optimal solution of VC-SNDP on ($G, r$).
\end{theorem}

We note a small but important difference between the algorithm implied by theorem above, and the one earlier. We increase the $f$ by a factor of $2$ which is 
needed for the proof below.

\begin{figure}[!ht]
    \centering
    \phantomsection
    \begin{tcolorbox}[colframe=gray!50!black, colback=gray!10!white,          coltitle=black, rounded corners, fonttitle=\bfseries, width=.7\linewidth]
    \uline{\textbf{\large \textsf{VC-SNDP-LP}}\;\;$\langle\mkern-4mu\langle\textsf{\it Input: } (G = (V, E), w, h)\rangle\mkern-4mu\rangle$}
    \begin{align*}
        \min & \sum_{e \in E} \vx(e) w(e)\\
        \text{s.t.} & \sum_{uv\in E\; :\; u \in S, v\in V\setminus S^+} \vx(e) \geq h(\hat{S})
        &&\forall S, S^+ \subseteq V, \text{s.t. } S\subseteq S^+\\
        & \vx(e)  \geq 0 &&\forall e \in E
    \end{align*}
    \end{tcolorbox}
    \caption{Biset-based LP relaxation of VC-SNDP}
    \label{fig:VC-SNDP-LP}
\end{figure}

\begin{proof}
    Let $H = H_1 \uplus \cdots \uplus H_T$, where $T = \epsilon^{-1} \log W$, be the output of Algorithm~\ref{alg:FT-spanner-weighted} with $f = (2t-2)(2k-1)$. Recall that for every $1\le j\le T$, $H_j$ is a $\big((2t-2)(2k-1)\big)$-VFT $(2t-1)$-spanner for the edges in $G$ with weights in $B_j \coloneqq ((1+\epsilon)^{j-1}, (1+\epsilon)^j]$.
     
    \medskip    
    Let $\opt \subseteq E$ be an optimal solution for the VC-SNDP instance on $G$. Starting from $\opt$, we construct a feasible fractional solution $\vx$ supported only on $H$, for the standard (bi)cut-based relaxation (i.e.,~\hyperref[fig:VC-SNDP-LP]{VC-SNDP-LP}) of the VC-SNDP instance, with cost at most $O(t) \cdot w(\opt)$. 
    
    In~\hyperref[fig:VC-SNDP-LP]{VC-SNDP-LP}, $h:2^V \times 2^V \rightarrow\{0,1,\cdots, k\}$ is defined as $h(\hat{S})\coloneqq \max(0, \max_{v\in S, u\in V\setminus S^+} r(u,v) - |S^+ \setminus S|)$, when $S \subseteq S^+$. Otherwise, $h(\hat{S})=0$. In particular, observe that for every biset $\hat{S}$ with $h(\hat{S}) > 0$, $h(\hat{S}) + |S^+\setminus S| \le k$.
    
    Then, the fractional solution $\vx$ is constructed from $\opt$ as follows (in Figure~\ref{fig:fractional-vc-solution-construction}). 
    \begin{figure}[!ht]
    \centering
    \phantomsection
    \begin{tcolorbox}[colframe=gray!50!black, colback=white, coltitle=black, rounded corners, fonttitle=\bfseries, width=.85\linewidth]
    \uline{\textbf{\large \textsf{Construction of $\vx$ from $\opt$}}}
    \begin{itemize}[leftmargin=*]
        \setlength{\itemsep}{0pt}
        \setlength{\parskip}{0pt} \setlength{\parsep}{0pt}  
        \item[] {\bf Initialize} $\vx$ as an all-zero vector, i.e., $\vx(e) = 0$ for all $e \in E$. 
        \item[] {\bf for each} $e =(u,v) \in \opt$
        \begin{itemize}[leftmargin=*]
            \setlength{\itemsep}{0pt}
            \setlength{\parskip}{0pt} \setlength{\parsep}{0pt}
            \item[] {\bf let} $B_j$ denote the weight class to which the weight of $e$ belongs. 
            \item[] {\bf if} $e\in H_j$ {\bf then} $\vx(e) \leftarrow 1$
            \item[] {\bf else} 
            \begin{itemize}[leftmargin=*]
                \setlength{\itemsep}{0pt}
                \setlength{\parskip}{0pt} \setlength{\parsep}{0pt}
                \item[] {\bf let} $P_1, \cdots, P_{2k}$ be vertex-disjoint $uv$-paths in $H_j$, each of length at most $t$
                \item[] {\bf for each $e' \in P_1 \cup \cdots \cup P_{2k}$, $\vx(e') \leftarrow \min(1, \vx(e') + 1/k)$}
            \end{itemize}
        \end{itemize}
    \end{itemize}
    \end{tcolorbox}
    \caption{Construction of the fractional solution $\vx$ of~\hyperref[fig:VC-SNDP-LP]{VC-SNDP-LP} from $\opt$}
    \label{fig:fractional-vc-solution-construction}
    \end{figure}
    Note that our algorithm does not need to explicitly construct $\vx$; this is only for analysis purposes.
    
    \medskip
    First, we prove the feasibility of $\vx$ for~\hyperref[fig:VC-SNDP-LP]{VC-SNDP-LP}($G, w, h$). Consider a biset $\hat{S}$ with $h(\hat{S})>0$. Note that for this to hold, it requires that $|S^+\setminus S| < r(\hat{S})$.
    In the optimal solution $\opt$, there are at least $h(\hat{S})$ edges in $\delta_G(\hat{S})$.  
    Let $L_1$ and $L_0$ respectively denote the number of edges in $\delta_{\opt}(\hat{S})$ that belong to $H$ and those that do not; i.e., $L_1 = |\delta_{\opt}(\hat{S}) \cap H|$ and $L_0 = |\delta_{\opt}(\hat{S}) \setminus H|$.
    Note that since $\vx(e) = 1$ for every $e \in H$, if $L_1 \ge h(\hat{S})$ then the fractional solution $\vx$ satisfies the connectivity requirement of $\hat{S}$. 
    
    Next, consider the case in which $L_1 < h(\hat{S}) = k'$.     
    We select $k' - L_1$ edges from the edge set $\delta_{\opt}(\hat{S}) \setminus H$ and denote them as $e_1 \coloneqq (u_1, v_1), \dots, e_{k'-L_1} \coloneqq (u_{k'-L_1}, v_{k'-L_1})$. 
    Then, for each $i \le k'-L_1$, consider the $2k$ vertex-disjoint $(u_i, v_i)$-paths $P_1, \dots, P_{2k}$ in $H_j$, as shown in Observation~\ref{obs:vft-spanner-structure}, and used in the construction of $\vx$ (see Figure~\ref{fig:fractional-vc-solution-construction}). Since $P_1, \cdots, P_{2k}$ are vertex-disjoint and $u_iv_i$ is crossing the biset $\hat{S}$, at least 
    \begin{align*}
        2k - L_1 - |S^+ \setminus S|  > 2k - L_1 - |S^+ \setminus S| - (k' - L_1) &= 2k - (k' + |S^+\setminus S|) \\
        &= 2k - (h(\hat{S}) + |S^+\setminus S|) \ge k 
    \end{align*}
    of the paths must have an edge crossing $\hat{S}$ distinct from $\delta_{\opt}(\hat{S}) \cap H$, where the last inequality holds because for every biset $\hat{S}$ with $h(\hat{S}) >0$, $h(\hat{S}) + |S^+\setminus S| \le k$. 
    Without loss of generality, let us denote $k$ of these distinct edges by $e^i_1, \dots, e^i_{k}$. Then, by the construction of $\vx$, 
    \begin{align*}
        \sum_{e\in \delta_H(\hat{S})} \vx(e) 
        \ge 
        \sum_{e\in \delta_{\opt}(\hat{S}) \cap H} \vx(e) + \sum_{e\in \{e^i_j\}_{i\in [k' - L_1], j\in [k]}} \vx(e) \ge L_1 + (k' - L_1) \cdot k \cdot \frac{1}{k} = k' = h(\hat{S}).
    \end{align*}
    Note that the second inequality holds because $\delta_{\opt}(\hat{S})\cap H$ and $\{e^i_j\}_{i\in [k'-L_1], j\in [k]}$ are disjoint and no edge appears more than $k'-L_1$ times in the second summation (i.e., $\sum_{e\in \{e^i_j\}_{i\in [k'-L_1], j\in [k]}} \vx(e)$). Hence, for every edge $e' \in \{e^i_j\}_{i\in [L], j\in [k]}$, its $\vx$ value is at most $\min(1, c(e')/k) = c(e')/k$, because $c(e')$ which is defined as the number of times $e'$ appears in the collection $\{e^i_j\big\}_{i \in [k'-L_1], j \in [k]}$, is at most $k' -L_1\le k' \le k$. So, $\vx$ is a feasible solution for~\hyperref[fig:VC-SNDP-LP]{VC-SNDP-LP}($G, w, h$). 
    
    \medskip
    For the cost analysis, note that for every edge $e = (u, v) \in \opt$ in weight class $B_j$, either $\vx(e)=1$, or it contributes in increasing the $\vx$ values at most $2k\cdot (2t-1) \cdot \frac{1}{k} = 2\cdot (2t-1)$ units on edges whose weight belong to $B_j$. Therefore, $w(\vx) \le 2\cdot (2t-1) \cdot (1 + \epsilon) \cdot w(\opt) = 4t\cdot w(\opt)$.
\end{proof}

\begin{corollary}\label{cor:stream-vc-sndp-fractional}
    There exists an algorithm for VC-SNDP with edge weights $w: E \rightarrow \{0, 1, \dots, W\}$ and a maximum connectivity requirement $k$, in insertion-only streams, that uses $\tilde{O}(k^{1-1/t}\cdot n^{1+1/t})$ space and outputs a $(4t \beta)$-approximation where $\beta$ is the integrality gap of the cut-based LP relaxation.
\end{corollary}

\paragraph{Implications for VC-SNDP and special cases:} 
The main advantage of the fractional analysis is for those cases where the integrality gap of the LP is small. 
We point out some of those cases and note that this is particularly useful for EC-SNDP and ELC-SNDP, which we discuss later.

\begin{itemize}
    \item For $\{0,1,2\}$-VC-SNDP there is an algorithm that uses $\tilde{O}(n^{1+1/t})$ space and 
    outputs a $8t$-approximate solution. This follows via the known integrality gap of $2$ for these instances \cite{fleischer2006iterative,CheriyanVV06}.
    \item For $k$-VCSS there is an algorithm that uses $\tilde{O}(k^{1-1/t}\cdot n^{1+1/t})$ space and outputs a $(16+\eps)t$-approximate solution when $n$ is sufficiently large compared to $k$. This follows from the known the integrality gap results for $k$-VCSS from \cite{Nutov22,CheriyanV14,fukunaga2015iterative}. This also implies a similar result for the connectivity augmentation problem $k$-VC-CAP.
    \item For finding the cheapest $k$ vertex-disjoint $s$-$t$ paths, there is an algorithm that uses $\tilde{O}(k^{1-1/t}\cdot n^{1+1/t})$ space and outputs a $4t$-approximate solution. This follows from the fact that the flow-LP is optimal for $s$-$t$ disjoint paths.
\end{itemize}
We believe that the regime of $k$ being small compared to $n$ is the main interest in the streaming setting. 
For large values of $k$, $O(t\log(\frac{n}{n-k}))$ and $O(t\log k \cdot \log(\frac{n}{n-k}))$ approximation bounds can be derived for $k$-VC-CAP and $k$-VCSS, respectively, via known integrality gaps (see~\cite{nutov2018improved} and~\cite{nutov2014approximating}). We omit formal statements in this version.

\subsection{EC-SNDP and ELC-SNDP}
\label{sec:ec-sndp}
The proof technique that we outlined for VC-SNDP applies very broadly and also hold for EC-SNDP and ELC-SNDP. We state below the theorem for EC-SNDP
that results from the analysis with respect to the fractional solution. 

\begin{theorem}\label{thm:FT-spanner-EC-SNDP-frational}
    Let $H$ be the EFT spanner of a weighted graph $G$ as constructed in Algorithm~\ref{alg:FT-spanner-weighted} with parameters $(t, f = (2t-1)(2k-1), \eps = 1/(2t-1))$. Then, the weight of an optimal fractional solution of EC-SNDP on ($H, r$) is within a $(4t)$-factor of the weight of an optimal solution of EC-SNDP on ($G, r$). 
\end{theorem}

We omit the proof since it follows the same outline as that for VC-SNDP.

\begin{corollary}\label{cor:EC-SNDP-stream-improved}
    There exists a streaming algorithm for EC-SNDP with edge weights $w: E \rightarrow \{0, 1, \dots, W\}$ and a maximum connectivity requirement $k$, in insertion-only streams, that uses $\tilde{O}\big(k^{1/2-1/(2t)} \cdot n^{1+1/t} + kn\big)$ space and outputs a $(8t)$-approximate solution.
\end{corollary}
\begin{proof}
The proof follows from Theorem~\ref{thm:FT-spanner-EC-SNDP-frational} which shows that Algorithm~\ref{alg:generic-stream-network-design} with $(t, f = (2t-1)(2k-1), \eps = 1/(2t-1))$, returns an EFT $2t$-spanner $H$ of size $\tilde{O}(k^{1/2-1/(2t)} \cdot n^{1+1/t} + kn)$ that contains a $(4t)$-approximate fractional solution for EC-SNDP on $G$. 

The integrality gap of {EC-SNDP-LP} is $2$ (cf.~\cite{jain2001factor}). 
Once the stream terminates, we perform an exhaustive search, i.e., enumerating all possible solutions on $H$ an output a $(8t)$-approximate solution.
Since the algorithm's only space consumption is for storing the constructed EFT spanner, it requires $\tilde{O}\big(k^{1/2-1/(2t)} \cdot n^{1+1/t} + kn\big)$ space, for odd $t$.
\end{proof}
Notably, in the setting where $W = n^{O(\log n)}$, Corollary~\ref{cor:EC-SNDP-stream-improved} improves upon the algorithm of~\cite{JKMV24}, which outputs an $O(t \cdot \log k)$-approximation using $\tilde{O}(k \cdot n^{1+1/t})$ space. More importantly, Algorithm~\ref{alg:generic-stream-network-design} is a more natural and generic approach for EC-SNDP compared to the analysis in~\cite{JKMV24}.

One can derive the following theorem for ELC-SNDP in a similar fashion. For this we rely on the fact that the integrality gap of the biset based LP for ELC-SNDP is
$2$ \cite{fleischer2006iterative,CheriyanVV06}. We omit the formal proof since it is very similar to the ones above for VC-SNDP and EC-SNDP.

\begin{theorem}\label{thm:ELC-SNDP-stream}
    There exists a streaming algorithm for ELC-SNDP with edge weights $w: E \rightarrow \{0, 1, \dots, W\}$ and a maximum connectivity requirement $k$, in insertion-only streams, that uses $\tilde{O}\big(k^{1-1/t} \cdot n^{1+1/t}\big)$ space and outputs a $(8t)$-approximate solution. 
\end{theorem}

\subsection{Efficient Streaming Algorithms}
\label{sec:polytime}
The general framework and corresponding streaming algorithms discussed so far aim to optimize the tradeoff between space usage and approximation ratio. However, these algorithms, as described, may run in exponential time. In this section we outline approaches to obtain polynomial-time algorithms and the necessary space/approximation tradeoffs. We consider the two parts of the framework, (1) the streaming algorithm for storing a fault-tolerant spanner and (2) the algorithm at the end of the stream, separately. 

\paragraph{Streaming Algorithm:} The standard greedy algorithm for $f$-FT spanners runs in time $O(n^f)$ and this is not computationally efficient for large values of $f$. The study of polynomial-time constructions of such spanners is an active research question~\cite{dinitz2020efficient,bodwin2021optimal,bodwin2022partially}. In particular, the following polynomial-time constructions can be implemented in the streaming setting. In all the methods described, the key modification is to replace the na\"ive $O(n^f)$-time $(u, v)$ test (i.e., checking whether there exists a set of f edges or vertices whose removal increases the distance between $u$ and $v$ in the remaining subgraph to at least $2t$) with a polynomial-time $(u,v)$ test.

\begin{itemize}
    \item{\bf VFT:} Bodwin, Dinitiz and Robelle~\cite{bodwin2021optimal} proposed a randomized implementation of the greedy approach for $f$-VFT $(2t-1)$-spanners. In their method, when testing whether to add an edge $(u,v)$, the algorithm randomly samples $\Theta(\log n)$  induced subgraphs of the current spanner. Each subgraph includes $u,v$, and each of the remaining vertices is included independently with probability $1/(2f)$. Then, the edge $(u,v)$ is added to the spanner if, in at least $1/4$ fraction of the sampled subgraphs, the distance between $u$ and $v$ exceeds $2t-1$. They show that this algorithm succeeds with high probability. It is simple to see that this can be implemented in the streaming setting in space that is proportional to the size of the spanner. Hence, it is possible to construct an $f$-VFT $(2t-1)$-spanner of optimal size, i.e., $\tilde{O}(f^{1-1/t} n^{1-1/t})$, in polynomial time in the streaming model, that succeeds with high probability. 
    
    \item{\bf EFT:} Dinitz and Robelle~\cite{dinitz2020efficient} provided a simple polynomial-time construction of $f$-EFT $(2t-1)$-spanners with slightly increased size complexity, i.e., $O(t f^{1-1/t} n^{1+t})$, which can be implemented efficiently in the streaming model too. Specifically, in their approach, when testing whether to add an edge $(u,v)$, they perform $f$ iterations as follows: starting with an initially empty set of edges $F$, in each iteration they attempt to find a path of length at most $(2t-1)$ between $u$ and $v$ in $S\setminus F$ and add it to $F$, where $S$ is the current spanner. If such a path exists, it is added to $F$. If, in any iteration, no such path is found in $S \setminus F$, then the edge $(u, v)$ is added to the spanner $S$. Otherwise, it is not added. 
    Later, Bodwin, Dinitz, and Robelle~\cite{bodwin2022partially} provided an improved analysis of the same (polynomial time) algorithm, achieving size bounds only slightly worse than those of the greedy approach with an exponential-time $(u,v)$ test. Their result constructs $f$-EFT $(2t-1)$-spanners of size 
    $O(t^{2.5} f^{1/2 - 1/(2t)} n^{1 + 1/t} + t^2 f n)$ for odd $t$. 
    For even $t$, the spanner size is 
    $O(t^{2.5} f^{1/2} n^{1 + 1/t} + t^2 f n)$.
 
\end{itemize}

\paragraph{Calculating Feasible Solution:} In all the described algorithms, we perform an exponential time exhaustive search, by enumerating all possible solutions, on the constructed spanner to find an optimal solution after the stream terminates. 
The spanner effectively serves as a coreset for the problem. Since most of the problems considered in this paper are NP-hard, one could instead run known polynomial time approximation algorithms in the offline setting. We discuss some such approximation algorithms of interest.

The best known approximation for EC-SNDP is a 2-approximation via iterated rounding \cite{jain2001factor}. However, this requires exactly solving the cut-based LP; it is unclear if this can be done in linear space, especially given that this LP is exponentially sized. There is a primal-dual 2-algorithm for the \emph{connectivity augmentation} problem \cite{WGMV93}; this can be implemented in linear space. While the algorithm does not directly solve the LP, the analysis shows that it is a 2-approximation with respect to the optimal fractional solution. This can be repeatedly applied to obtain an $O(\log k)$-approximation (where $k$ is the maximum connectivity requirement) for EC-SNDP \cite{GoemansGPSTW94}. For vertex-connectivity, The best known approximation for VC-SNDP employs a reduction to element connectivity. Given an $\alpha$-approximation for ELC-SNDP that uses $f(m)$ space, one can obtain a $O(\alpha k^3 \log n)$-approximation in $O(m + f(m))$ space \cite{chuzhoy2009k}. Some algorithms for special cases of VC-SNDP (such as single-source, $k$-VC-CAP, etc.) may be modified to run in linear space; we omit details in this version for brevity and refer the reader to related work discussed in Section \ref{sec:rel-work}. 

Combining the two parts, we achieve the following efficient streaming algorithms, where $S$ is the space usage for the corresponding streaming problem (i.e. $S = k^{1 - 1/t} \cdot n^{1 + 1/t}$ for vertex-connectivity instances, and $S = k^{1/2 - 1/(2t)} \cdot n^{1 + 1/t} + kn$ for edge-connectivity instances):
\begin{itemize}
    \item{\bf Integral Solution Approach.} If there exists a polynomial time $\alpha$-approximation algorithm for the SNDP problem with space complexity $g(m)$, where $m$ is the size of the input graph, our framework returns an $O(\alpha t k)$-approximate solution in polynomial time, using $\tilde{O}(S + g(S))$ space.

    \item{\bf Fractional Soultion Approach.} If there exists a polynomial time algorithm for the SNDP problem that returns an {\em integral} solution within an $\alpha_{\mathrm{fr}}$-factor of the optimal fractional solution to the cut-based LP relaxation of the problem, with space complexity $g(m)$ (where $m$ is the size of the input graph), then our framework returns an $O(\alpha_{\mathrm{fr}} \cdot t)$-approximate solution in polynomial time, using $\tilde{O}(S + g(S))$ space.
\end{itemize}
\section{Vertex Connectivity Augmentation in Link-Arrival Model}
\label{sec:vc_link_arrival}

In this section we consider $k$-VC-CAP in the link arrival setting. Recall that in this problem, we are initially given an underlying $k$-vertex-connected graph $G = (V, E)$, while additional links $L \subseteq V \times V$ with edge weights $w: L \to \mathbb{R}_{\geq 0}$ arrive as a stream. For readability, we refer to edges of the given graph $G$ as ``edges'' and the incoming streaming edges $L$ as ``links''.  Our goal is to find the min-weight subset $L' \subseteq L$ such that $(V, E \cup L')$ is $(k+1)$-vertex-connected. Note that this model is easier than the fully streaming setting, so results in Section \ref{sec:generic} immediately apply here as well. We show that for $k = 1, 2$, we can obtain constant-factor approximations in near-linear space. 

For ease of notation, we write $k$-connected to mean $k$-vertex-connected. In Section \ref{sec:vc_1_to_2} we describe a simple algorithm for augmenting from 1 to 2 connectivity. This algorithm relies on the fact that the underlying 1-connected graph is a tree. Unfortunately, $2$-VC-CAP is significantly less straightforward, since 2-connected graphs do not share the same nice structural properties as trees. To circumvent this issue, we use a data structure called an SPQR tree, which is a tree-like decomposition of a 2-connected graph into its 3-connected components. We describe SPQR trees and their key properties, along with the augmentation algorithm from 2 to 3 connectivity, in Section \ref{sec:vc_2_to_3}. 

\begin{remark}
\label{rem:vc_link_polytime}
    For both $1$-VC-VAP and $2$-VC-CAP, the ``in-stream'' portion of our algorithms run in polynomial time. The approximation ratios we provide assume an exact computation of an optimal solution on the stored edges; this can be done via exhaustive search in linear space, but uses exponential time. 
    For a polynomial time algorithm, one could instead use an approximation algorithm at the end of the stream. $1$-VC-CAP admits a 3-approximation via primal-dual that can be implemented in linear space \cite{ravi2002erratum}. $2$-VC-CAP also admits a 2-approximation \cite{ADNP99}; this uses as a subroutine an algorithm for computing a minimum cost \emph{directed} graph that contains $3$ node-disjoint paths from a specified root to all other vertices in the graph. This problem is known to be polynomial-time solvable \cite{edmonds2003submodular}, but it is unclear if that can be done in linear space. However, one can obtain an $O(1)$-approximate solution by repeatedly augmenting connectivity from the root via primal-dual methods \cite{Frank_1999}. Thus one can obtain an $O(1)$-approximation for 2-VC-CAP in linear space.
\end{remark}

\subsection{One-to-Two Augmentation}\label{sec:vc_1_to_2}
In this section, we prove the following theorem:

\begin{theorem}\label{thm:vc_1_to_2_main}
    There exists a streaming algorithm for $1$-VC-CAP with edge weights $w: E \to [1, W]$, in an insertion-only stream, that uses $O(n \eps^{-1} \log W)$ space and outputs a $(3 + \eps)$-approximate solution. 
\end{theorem}

We fix a 1-connected graph $G = (V, E)$. We can assume without loss of generality that $G$ is a tree; if not, we can fix a spanning tree of $G$ as the underlying graph, and consider all remaining edges as $0$-weight links in the stream. It is easy to verify that this does not change the problem. We fix an arbitrary root $r$ of the tree $G$. For each $u \in V$, we let $G_u$ denote the subtree of $G$ rooted at $u$, and we let $C(u)$ denote the set of children of $u$. For $u, v \in V$, we let $LCA(u,v)$ denote the lowest common ancestor of $u$ and $v$ in the tree $G$; this is the vertex $w$ furthest from the root such that $u, v \in G_w$. We overload notation and write $LCA(e)$ for an edge $e$ to be the LCA of its endpoints. We let $d_G(u,v)$ denote the number of edges in the unique tree path between $u$ and $v$. For any edge set $E' \subseteq E \cup L$ and any vertex sets $A, B \subseteq V$, we let $E'[A, B]$ denote the set of edges in $E'$ with one endpoint in $A$ and the other in $B$.
We use the following lemma as a subroutine:

\begin{lemma}\label{lem:mst_streaming}\cite{mcgregor2014graph}
    Given any set of nodes $V'$ with links $L \subseteq V' \times V'$ appearing in an insertion-only stream, one can store a minimum spanning tree on $V'$ using $O(|V'|)$ memory space. 
\end{lemma}

\subsubsection{The Streaming Algorithm}

The streaming algorithm is as follows. For each vertex $u \in V$ and each weight bucket $[(1+\eps)^j, (1+\eps)^{j+1})$, we store the link $uv \in L$ in this weight bucket with $LCA(u,v)$ closest to the root; these will be stored in the dictionaries $L_u$ defined in Algorithm \ref{alg:vc_1_to_2}. Furthermore, for each $u \in G$, we consider a contracted graph with $C(u)$ nodes, where each $G_v$ for $v \in C(u)$ is contracted into a node. We maintain an MST on this contracted graph. The details are formalized in Algorithm \ref{alg:vc_1_to_2} below. 

\begin{algorithm}[H]
    \KwIn{Weighted tree $G = (V, E)$ with root $r \in V$ and edge weights $w: E\rightarrow [1, W]$.}
    \caption{The streaming algorithm for 1-to-2 VCSS augmentation}
    \label{alg:vc_1_to_2}
    \tcc{\sc {Preprocessing}:}
    \For{$u \in V$}{
        Initialize an empty dictionary $L_u$ \\
        Construct $C'(u)$: for each $v \in C(u)$, contract $G_u$ into one ``supernode'' and add it to $C'(u)$ \\
        Define the graph $T'_u$ with vertex set $C'(u)$ and edge set $\emptyset$
    }

    \tcc{\sc {In Stream}:}
    \For {$e = uv \in L$ in the stream}{
        Let $j$ be weight class such that $w(e) \in [(1+\eps)^j, (1+\eps)^{j+1})$\\
        \If{$L_u(j)$ is undefined \textbf{or} $L_u(j) = uv'$ such that $d_G(r, LCA(u,v)) < d_G(r, LCA(u, v'))$}{
            $L_u(j) \gets e$            
        }
        \If{$L_v(j)$ is undefined \textbf{or} $L_v(j) = u'v$ such that $d_G(r, LCA(u,v)) < d_G(r, LCA(u', v))$}{
            $L_v(j) \gets e$            
        }
        Update MST $T_x'$ for $x = LCA(u,v)$ with link $uv$ using Lemma \ref{lem:mst_streaming}
    }
    \tcc{\sc {Postprocessing}:}
    $F = \cup_{u \in V} (E(T_u') \cup L_u)$ \\
    \Return An optimal solution on link set $F$
\end{algorithm}

\begin{lemma}\label{lem:vc_1_to_2_space}
    The number of links stored in Algorithm \ref{alg:vc_1_to_2} is $O(n \eps^{-1} \log W)$. 
\end{lemma}
\begin{proof}
    The set of links stored in Algorithm \ref{alg:vc_1_to_2} is exactly $F$. For each vertex $u \in V$, we store one link in $L_u$ per weight class. Thus $|\cup_{u \in V} L_u| = \sum_{u \in V} O(\eps^{-1} \log W) = O(n \eps^{-1} \log W)$. We bound the number of links in the spanning trees by 
    $|\cup_{u \in V} E(T_v')| \leq \sum_{u \in V}|C(u) - 1| \leq \sum_{u \in V} \deg(u) = 2|E|$. Since we assume $G$ is a tree, $|E| = n-1$, so $|\cup_{u \in V} E(T_v')| = O(n)$. 
\end{proof}

\subsubsection{Bounding Approximation Ratio}
For the rest of this section, we bound the approximation ratio by proving the following lemma. 

\begin{lemma}\label{lem:vc_1_to_2_approx}
    Let $F$ be the set of links stored in Algorithm \ref{alg:vc_1_to_2}. The weight of an optimal solution to $1$-VC-CAP on $(V, F)$ is at most $(3+\eps)$ times the optimal solution to $1$-VC-CAP $G = (V, E)$. 
\end{lemma}

Fix an optimal solution $\opt$ on $G = (V, E)$. 
To prove Lemma \ref{lem:vc_1_to_2_approx}, we provide an algorithm to construct a feasible solution $\sol \subseteq F$ such that $w(\sol) \leq (3 + \eps) w(\opt)$. Note that this algorithm is for analysis purposes only, as it requires knowledge of an optimal solution $\opt$ to the instance on $G = (V, E)$. 

\begin{algorithm}[H]
\caption{Construction of a solution in $F$ from $\opt$.}
\label{alg:vc_1_to_2_analysis}
    $\sol \gets \emptyset$ \\
    \For{$e = uv \in \opt$}{
        Let $j$ be weight class such that $w(e) \in [(1+\eps)^j, (1+\eps)^{j+1})$ \\
        $\sol \gets \sol \cup \{L_u(j), L_v(j)\}$
    }
    \For{$u \in V$}{
        \For{$v \in C(u)$}{
            \If{$\exists e \in \opt[G_v, G \setminus G_u]$}{
                Mark $v$ as ``good''
            }
        }
        Construct $C''(u)$ from $C'(u)$ (as defined in Algorithm \ref{alg:vc_1_to_2}) by contracting all supernodes corresponding to a ``good'' vertex $v \in C(u)$ into one ``good'' supernode \\
        $T''_u \subseteq T'_u \gets$ MST on $C''(u)$\\
        $\sol \gets \sol \cup T''_u$
    }
    \Return $\sol$
\end{algorithm}

We fix $\sol$ to be the set of links given by Algorithm \ref{alg:vc_1_to_2_analysis}. 

\begin{lemma}
\label{lem:vc_1_to_2_feasibility}
    $(V, E \cup \sol)$ is a 2-connected graph.
\end{lemma}
\begin{proof}
    We want to show that for all $a \in V$, $(V, E \cup \sol) \setminus \{a\}$ is connected. Fix $a \in V$. Since $\opt$ is feasible, it suffices to show that for all $uv \in \opt$ with $u, v \neq a$, there exists a $u$-$v$ path in $E \cup \sol$ that does not use $a$.

    Fix $uv \in \opt$. Consider the unique $u$-$v$ path in $E$ (recall that $G$ is a tree). If this path does not contain $a$, then we are done. Thus we assume $a$ is on the $u$-$v$ tree path. We case on whether or not $a$ is the $LCA$ of $u$ and $v$. Let $b = LCA(u,v)$. Let $u', v' \in C(b)$ be the children of $b$ such that $u \in G_{u'}$ and $v \in G_{v'}$.

    \begin{itemize}[leftmargin=*]
        \item \textbf{Case 1:} First, suppose $a = b$. If $u'$ or $v'$ are not marked as ``good'' in Algorithm \ref{alg:vc_1_to_2_analysis}, then $u'$ and $v'$ correspond to separate supernodes of $C''(a)$ and thus $G_{u'}$ and $G_{v'}$ must be connected via $T''_{a}$. Both $G_{u'}$ and $G_{v'}$ remain connected despite the deletion of $a$ since $u', v'$ are children of $a$; thus in this case, $u$ and $v$ are connected via $T''_a$. See Figure \ref{fig:vc_1_to_2_case1} for reference. Note that this case is the only one in which we use the contracted MST $T_{a}''$.
        
        Suppose instead that \emph{both} $u'$ and $v'$ are marked as ``good'', then we will show that they are connected via $G \setminus G_a$ (see Figure \ref{fig:vc_1_to_2_case1}). Since $G_{u'}$ is good, there must be some link $e_u \in \opt[G_{u'}, G \setminus G_a]$. Let $u''$ be the vertex incident to $e_u$ in $G_{u'}$, and let $j$ be the weight class of $e_u$. Then $L_{u''}(j)$ must have an endpoint in $G \setminus G_a$, since 
        \[d_G(r, LCA(L_{u''}(j))) \leq d_G(r, LCA(e_u)) < d_G(r, a).\]
        Furthermore, $L_{u''}(j)$ is included in $\sol$ in Algorithm \ref{alg:vc_1_to_2_analysis}, since $e_u \in \opt$. 
        Similarly, there must be some $v'' \in G_{v'}$ such that $L_{v''} \cap \sol$ contains a link with one endpoint in $G \setminus G_a$. Since $G \setminus G_a$ remains connected despite the deletion of $a$, we obtain our desired $u$-$v$ path. 
        \item \textbf{Case 2:} If $a \neq b$, then $a$ is either in $G_{u'}$ or in $G_{v'}$. Suppose without loss of generality that $a \in G_{u'}$; the other case is analogous.
        Let $j$ be the weight class of $uv$. Then 
        \[d_G(r, LCA(L_{u}(j))) \leq d_G(r, b) < d_G(r, a);\]
        the first inequality holds since $v \in G \setminus G_b$. Thus $L_{u}(j)$ has one endpoint in $G \setminus G_a$. Furthermore, $L_u(j) \in \sol$ since $uv \in \opt$. Since $G \setminus G_a$ remains connected despite the deletion of $a$, this gives us our desired $u$-$v$ path in $E \cup \sol$. See Figure \ref{fig:vc_1_to_2_case2} for reference.
    \end{itemize}
\end{proof}

\begin{figure}
    \centering
    \begin{subfigure}{0.45\textwidth}
        \centering
        \includegraphics[scale=0.6]{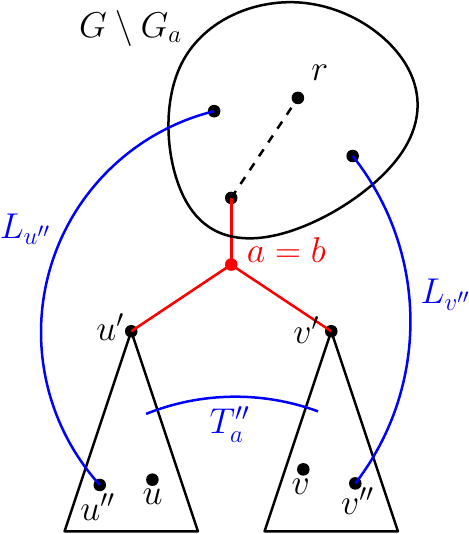}
        \caption{Example of Case 1. If $u'$ or $v'$ are not marked ``good'', $u$ and $v$ are connected via the link in $T_{a}''$ shown. Else, they are connected using the links in $L_{u''}$ and in $L_{v''}$.}
        \label{fig:vc_1_to_2_case1}
    \end{subfigure}%
    \hfill
    \begin{subfigure}{0.45\textwidth}
        \centering
        \includegraphics[scale=0.6]{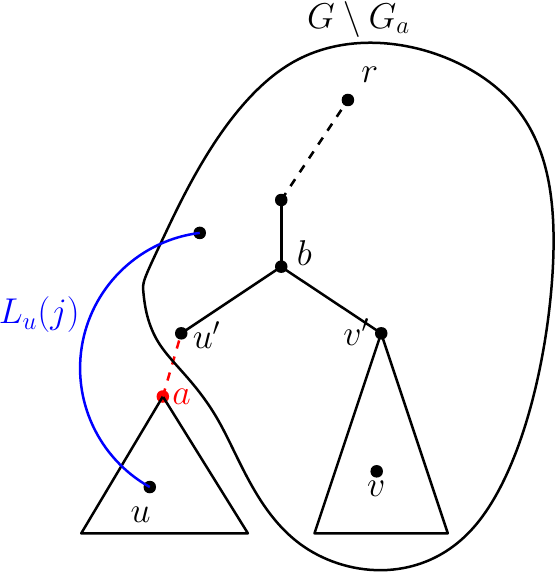}
        \caption{Example of Case 2. Here $u \in G_a$, $v \in G \setminus G_a$, and $a \in G_{u'}$.}
        \label{fig:vc_1_to_2_case2}
    \end{subfigure}
    \caption{Examples for Lemma \ref{lem:vc_1_to_2_feasibility}. $a$ represents the deleted vertex and $b = LCA(u,v)$. Dashed lines represent paths while solid edges represent edges.}
    \label{fig:vc_1_to_2_eg}
\end{figure}

\begin{lemma}
\label{lem:vc_1_to_2_weight}
    The link set $\sol$ has weight at most $(3 + \eps)w(\opt)$. 
\end{lemma}
\begin{proof}
    In the first part of Algorithm \ref{alg:vc_1_to_2_analysis}, for each link $uv \in \opt$, we add two links to $\sol$, each of weight at most $(1+\eps)w(uv)$. Thus the total weight of links added to $\sol$ in the first part is at most $2(1+\eps)w(\opt)$. 

    We claim that the second part of the algorithm costs at most $\opt$. To show this, consider a partition of $\opt$ based on the LCA of its endpoints: let $E^*(a) = \{e \in \opt: LCA(e) = a\}$. For each $v \in C(a)$, there must be some path in $E \cup \opt$ from $G_v$ to $G \setminus G_a$ \emph{not} containing $a$; else $a$ would be a cut-vertex separating $G_v$ from the rest of the graph, contradicting feasibility of $\opt$. Thus this path must contain some edge $e \in \opt[G_a \setminus \{a\}, G \setminus G_a]$. By definition, the endpoint of $e$ in $G_a \setminus \{a\}$ must be in a ``good'' supernode of $C'(a)$. Thus all ``bad'' supernodes of $C'(a)$ are connected to at least one ``good'' supernode of $C'(a)$ in $\opt$. Furthermore, a minimal path in the contracted graph from a ``bad'' supernode to a ``good'' supernode only uses links in $E^*(a)$, since these links all must go between subtrees of $a$. Thus in $E^*(a)$, all ``bad'' supernodes are connected to at least one ``good'' supernode. In particular, $E^*(a)$ contains a spanning tree on the contracted graph with vertices $C''(a)$. Therefore 
    \[\sum_{a \in V} w(E(T''_a)) \leq \sum_{a \in V} w(E^*(a)) \leq w(\opt).\]
    Thus the total weight of $\sol$ is at most $(3 + 2\eps) w(\opt)$. We can run the algorithm with $\eps/2$ to obtain the desired approximation ratio. 
\end{proof}

By Lemmas \ref{lem:vc_1_to_2_feasibility} and \ref{lem:vc_1_to_2_weight}, Algorithm \ref{alg:vc_1_to_2_analysis} provides a feasible solution to $1$-VC-CAP on $(V, F)$ with weight at most $(3 + \eps) w(\opt)$; this concludes the proof of Lemma \ref{lem:vc_1_to_2_approx}. This, combined with Lemma \ref{lem:vc_1_to_2_space}, concludes the proof of Theorem \ref{thm:vc_1_to_2_main}.

\subsection{Two-to-Three Augmentation}
\label{sec:vc_2_to_3}

In this section, we prove the following theorem:

\begin{theorem}\label{thm:vc_2_to_3_main}
    There exists a streaming algorithm for the $2$-VC-CAP problem with edge weights $w: E \to [1, W]$, in an insertion-only stream, that uses $O(n \eps^{-1} \log W)$ space and outputs a $(7 + \eps)$-approximate solution. 
\end{theorem}

Throughout this section, to avoid confusion with two-vertex cuts $\{u,v\}$, we will denote all edges and links as $uv$ or $(u,v)$. Following the notation in Section \ref{sec:vc_1_to_2}, for any edge set $E' \subseteq E \cup L$ and any vertex sets $A, B \subseteq V$, we let $E'[A, B]$ denote the set of edges in $E'$ with one endpoint in $A$ and the other in $B$. We start with some background on SPQR trees.

\subsubsection{SPQR Trees}

An SPQR tree is a data structure that gives a tree-like decomposition of a 2-connected graph into \emph{triconnected} components. SPQR trees were first formally defined in \cite{spqr-tree} although they were implicit in prior work \cite{di1989incremental}, and the ideas build heavily on work in \cite{hopcroft_tarjan_spqr}.
We start with some definitions and notation. Note that the following terms are defined on \emph{multigraphs} (several parallel edges between a pair of nodes are allowed), even though we assume the input graph $G$ is simple. This is because the construction of SPQR trees follows a recursive procedure that may introduce parallel edges. Let $G = (V, E)$ be a 2-connected graph.

\begin{definition}
    A \emph{separation pair} $\{a, b\}$ is a pair of vertices $a, b \in V$ such that at least one of the following hold:
    \begin{enumerate}
        \item $G \setminus \{a,b\}$ has at least two connected components.
        \item $G$ contains at least two parallel edges $ab$ \emph{and} $G \setminus \{a, b\} \neq \emptyset$. 
    \end{enumerate}
    We define corresponding \emph{separation classes} as follows:
    \begin{enumerate}
        \item If $G \setminus \{a,b\}$ has at least two connected components $C_1, \dots, C_k$, the separation classes are $E[C_i] \cup E[C_i, \{a, b\}]$ for each $i \in [k]$. If there are any edges of the form $ab$, we add an additional separation class $E'$ containing all such parallel edges between $a$ and $b$.
        \item If not, the separation classes are $E_1 = \{e: e = ab\}$ (all parallel edges between $a$ and $b$) and $E_2 = E(G) \setminus E_1$ (all edges except parallel edges between $a$ and $b$). 
    \end{enumerate}
    Note that the separation classes partition the edge set $E(G)$.
    See Figure \ref{fig:vc_2_to_3_sep-class} for an example.
\end{definition}

\begin{figure}
    \centering
    \begin{subfigure}{0.45\textwidth}
        \centering
        \includegraphics[width=0.7\linewidth]{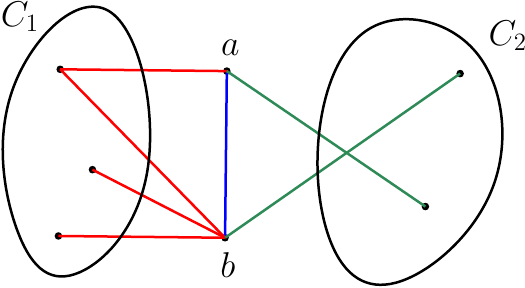}
        \caption{$G \setminus \{a, b\}$ has two components: $C_1, C_2$. There are three separation classes: red edges within $C_1$ and between $C_1$ and $\{a, b\}$, green edges within $C_2$ and between $C_2$ and $\{a, b\}$, and the blue edge $ab$.}
    \end{subfigure}%
    \hfill
    \begin{subfigure}{0.45\textwidth}
        \centering
        \includegraphics[width=0.45\linewidth]{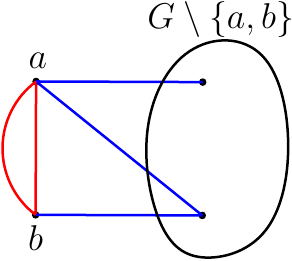}
        \caption{$G \setminus \{a, b\}$ has only one component, but there are at least two parallel edges between $a$ and $b$, and $G \setminus \{a, b\}$ is non-empty. There are two separation classes: red parallel edges between $a$ and $b$, and all other edges in blue.}
    \end{subfigure}
    \caption{Example of separation pair $\{a, b\}$ and corresponding separation classes shown with different colors.}
    \label{fig:vc_2_to_3_sep-class}
\end{figure}

Next, we define two operations: ``split'' and ``merge''.
\begin{definition}
    The \emph{split} operation on a 2-connected graph $G = (V, E)$ with $|E| \geq 4$ and separation pair $\{a, b\}$ is defined as follows. Let $E_1, \dots, E_k$ be the separation classes of $\{a,b\}$. For a subset of indices $I \subseteq [k]$, let $E_I' = \cup_{i \in I} E_i$ and $E_I'' = \cup_{i \notin I} E_i$. Choose $I$ arbitrarily such that $|E_I'| \geq 2, |E_I''| \geq 2$. We introduce a new \emph{virtual edge} $e = ab$, and define two new multigraphs:
    $G' = (V(E_I'), E_I' \cup \{e\}), G'' = (V(E_I''), E_I'' \cup \{e\})$. The \emph{split} operation replaces $G$ with \emph{split graphs} $G', G''$. 
\end{definition}

It is easy to verify that for a 2-connected graph $G$, the split graphs $G', G''$ (constructed using any separation pair) are also $2$-connected.

\begin{definition}
    The \emph{merge} operation merges two split graphs back to their original graph. Given split graphs $G' = (V', E')$ and $G'' = (V'', E'')$ with the same virtual edge $e$, the merge operation returns the graph $G = (V' \cup V'', (E' \cup E'') \setminus \{e\})$. 
\end{definition}

With these operations in hand, we are ready to define SPQR trees.

\begin{definition}
    The SPQR tree $T$ of a graph $G$ is a tree-like data structure in which each node of $T$ corresponds to some subgraph of $G$, with additional ``virtual'' edges. Formally, the SPQR tree $T$ is obtained from $G$ via the following recursive procedure:
    \begin{itemize}
        \item If $G$ has no separation pair, then $T$ is the singleton tree with one node $x$ and no edges. We write $G_x := G$ to denote the graph corresponding to tree node $x$. 
        \item If $G$ has a separation pair $\{a, b\}$, we split $G$ on $\{a, b\}$ to obtain split graphs $G', G''$. Let $e = ab$ be the corresponding virtual edge. We recursively construct SPQR trees $T', T''$ on $G', G''$. Let $x', x''$ be the nodes of $T', T''$ respectively containing edge $e$ -- this is well defined, since the split operation partitions the edges of $G$ and $e$ is contained in both $G'$ and $G''$. We let $T = T' \cup T'' \cup (x', x'')$; that is, we combine $T'$ and $T''$ via an edge between $x'$ and $x''$. We say the tree edge $(x', x'')$ is associated with the virtual edge $e$. 
    \end{itemize}
    Then, we postprocess $T$ as follows: if there are two adjacent tree nodes $x, y$ with $G_x, G_y$ both dipoles (i.e. $|V(G_x)| = |V(G_y)| = 2$), we merge $G_x$ and $G_y$. Similarly, if $x, y$ are adjacent tree nodes and both $G_x$ and $G_y$ are simple cycles, we merge $G_x$ and $G_y$. See Figure \ref{fig:vc_2_to_3_SPQR} for an example of the SPQR tree of a 2-connected graph. 
\end{definition}

\begin{figure}
    \centering
    \begin{subfigure}{0.3\textwidth}
        \centering
        \includegraphics[width=0.7\linewidth]{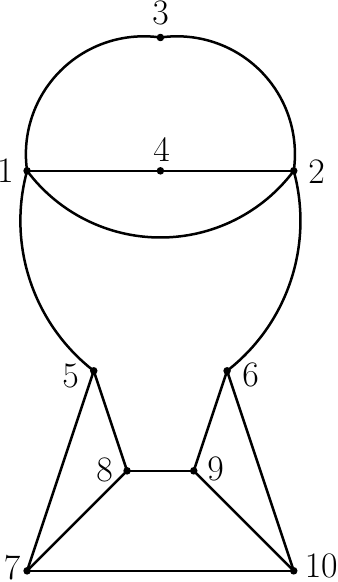}
        \caption{2-connected graph $G$.}
    \end{subfigure}%
    \hfill
    \begin{subfigure}{0.6\textwidth}
        \centering
        \includegraphics[width=0.8\linewidth]{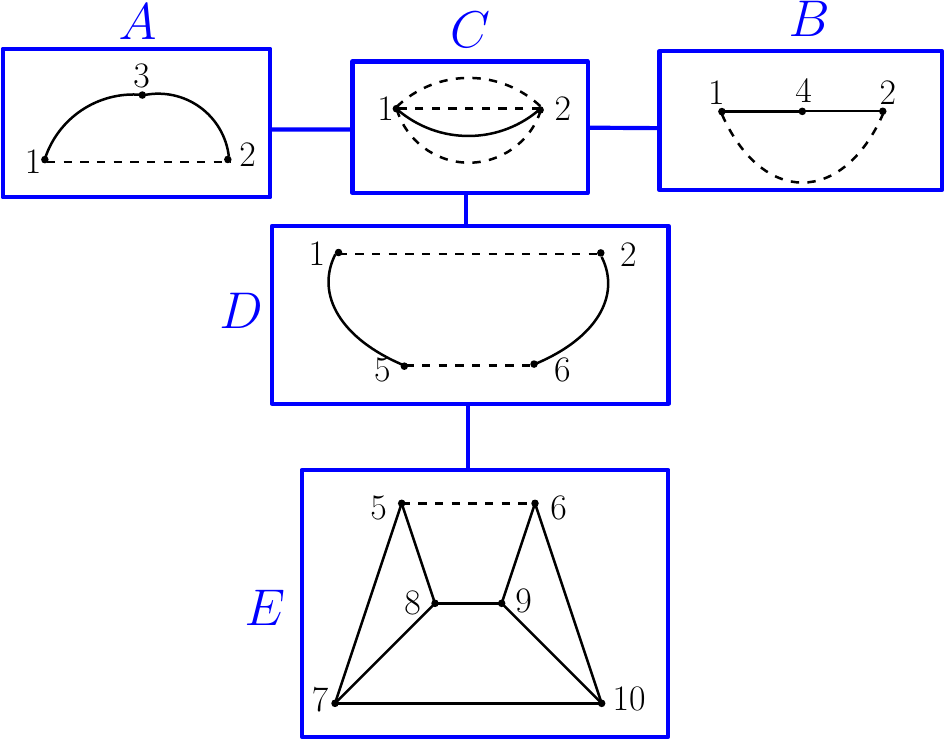}
        \caption{SPQR Tree $T$ with S-nodes $A, B, D$, P-node $C$, R-node $E$. Virtual edges are dashed.}
    \end{subfigure}
    \caption{Example of an SPQR tree $T$ constructed from 2-connected graph $G$. Suppose we root the tree at node $C$. As defined in Section \ref{sec:vc_2_to_3_algo}, $\parent(E) = \{5, 6\}$, and $A, B, D$ all have parent $\{1, 2\}$ (corresponding to different virtual edges). In this example, $h(1) = h(2) = C$, but $\ell(1), \ell(2)$ can be chosen arbitrarily from $\{A, B, D\}$. Vertices $3$ and $4$ are not included in any virtual edges; thus $h(3) = \ell(3) = A$ and $h(4) = \ell(4) = B$. The same holds for vertices $7,8,9,10$; $h$ and $\ell$ both map to $E$. Finally, $h(5) = h(6) = D$ and $\ell(5) = \ell(6) = E$.
    } 
    \label{fig:vc_2_to_3_SPQR}
\end{figure}

We use the following properties of SPQR trees, given mostly by \cite{hopcroft_tarjan_spqr,spqr-tree}. The linear-time implementation of Lemma \ref{lem:spqr_linear_time} was given by \cite{gutwenger2000linear}.

\begin{lemma}
    Each virtual edge is in exactly two tree nodes and is associated with a unique tree edge. Each edge in $E(G)$ is in exactly one tree node.
\end{lemma}

\begin{lemma}\label{lem:spqr_size}
    For any 2-connected graph $G$ with SPQR tree $T$, $\sum_{x \in V(T)} |E(G_x)| \leq 3|E(G)| - 6$.
\end{lemma}

\begin{lemma}\label{lem:spqr_linear_time}
    For any 2-connected graph $G$, there exists a linear-time algorithm to construct the SPQR tree $T$.
\end{lemma}

\begin{lemma}
    Let $x \in V(T)$. Then, $G_x$ is exactly one of the following:
    \begin{enumerate}
        \item a graph with exactly two vertices and three or more parallel edges between them;
        \item a simple cycle;
        \item a three-connected graph with at least 4 vertices.
    \end{enumerate}
    We call these P-nodes, S-nodes, and R-nodes respectively.
\end{lemma}

\begin{remark}
    The name ``SPQR'' tree comes from the different types of graphs that tree nodes can correspond to. A ``Q''-node has an associated graph that consists of a single edge: this is necessary for the trivial case where the input graph $G$ only has one edge. In some constructions of SPQR trees, there is a separate Q-node for each edge in $E(G)$; however, this is unnecessary if one distinguishes between real and virtual edges and thus will not be used in this paper.
\end{remark}

\begin{lemma}\label{lem:spqr_cuts}
    Let $\{a, b\}$ be a 2-cut of $G$. Then, one of the following must be true:
    \begin{itemize}
        \item $\{a, b\}$ is the vertex set of $G_x$ for a P-node $x$. The connected components of $G \setminus \{a, b\}$ correspond to the subtrees of $T \setminus \{x\}$.
        \item There exists a virtual edge $e = ab$ associated with a tree edge $xy$ such that at least one of $x$ and $y$ is an R-node; the other is either an R-node or an S-node. The connected components of $G \setminus \{a, b\}$ correspond to subtrees of $(V(T), E(T) \setminus xy)$. 
        \item $a$ and $b$ are non-adjacent nodes of a cycle $G_x$ for an S-node $x$. The connected components 
        of $G \setminus \{a,b\}$ are the subtrees corresponding to the two components of $G_x \setminus \{a, b\}$. 
    \end{itemize}
\end{lemma}

\subsubsection{The Streaming Algorithm}
\label{sec:vc_2_to_3_algo}

We fix a 2-connected graph $G = (V, E)$. We can assume without loss of generality that $G$ is an edge-minimal 2-vertex-connected graph, using the same reasoning as in the 1-to-2 augmentation setting. We start by constructing the SPQR tree $T$ of $G$; this can be done efficiently by Lemma \ref{lem:spqr_linear_time}. To avoid confusion, we refer to elements of $V(T)$ as ``nodes'' and elements of $V(G)$ as ``vertices''. We choose a root 
node $r$ of $T$. For each tree node $x \in V(T)$, we let $G_x$ denote its corresponding graph, we let $T_x$ denote the subtree of $T$ rooted at $x$, and let $C(x)$ denote the set of children of $x$. We define $\parent(x)$ as the virtual edge $ab$ of $G_x$ associated with the tree edge on the path from $x$ to $r$ (see Figure \ref{fig:vc_2_to_3_SPQR}); this is sometimes overloaded to refer to the set of vertices $\{a, b\}$. Recall from the construction of SPQR trees each vertex $u \in V(G)$ may have several ``copies'' in $T$; if $u$ is part of a separation pair, then it is duplicated in the ``split'' operation. It is not difficult to see that the set of all tree nodes containing a copy of $u$ forms a connected component of $T$. We denote by $h(u)$ and $\ell(u)$ the tree nodes containing $u$ that are the closest to and furthest from the root respectively (breaking ties arbitrarily); see Figure \ref{fig:vc_2_to_3_SPQR}. 

We provide a high-level description of the algorithm; this is formalized in Algorithm \ref{alg:vc_2_to_3}. By Lemma \ref{lem:spqr_cuts}, there are three possible type of 2-vertex cuts $\{a, b\}$ in $G$: (1) $\{a, b\}$ is the graph corresponding to a P-node, (2) $ab$ is a virtual edge associated with a tree edge incident to two nodes: one is an R-node and the other is an R-node or S-node, (3) $a$ and $b$ are non-adjacent nodes of a cycle $G_x$ for some S-node $x$. Intuitively, 
one can think of (1) as corresponding to a node being deleted in the tree, (2) 
corresponding to an edge being deleted in the tree, and (3) to handle connectivity within cycle nodes. 

\paragraph{``Tree'' Cuts:} To protect against 2-node cuts of type (1) and (2), we follow a similar approach to the 1-to-2 augmentation described in Section \ref{sec:vc_1_to_2}. The additional difficulty in this setting arises from the fact that $T$ may contain multiple copies of each vertex of $G$; thus a link $uv$ does not necessarily correspond to a unique ``tree link'' $xy \in V(T) \times V(T)$. Furthermore, a 2-vertex cut in $G$ may contain copies in several tree nodes. We handle this by strategically choosing to view a link $uv \in L$ as a tree link adjacent to either $h(u)$ or $\ell(u)$ (and $h(v)$ or $\ell(v)$) depending on which virtual edge failure we need to protect against. Specifically, we store the following:
\begin{itemize}
    \item for each tree node $x$, we store the link $uv$ with $x = h(u)$ minimizing $d_T(r, LCA(x, \ell(v)))$;
    \item for each P-node $x$, we construct the following contracted graph: for each child $y$ of $x$, we contract $(\cup_{z \in T_y} V(G_z)) \setminus V(G_x)$; this is the set of vertices of $G$ in the subtree $T_y$ \emph{not} including the two vertices in $G_x$. We store an MST on this contracted graph.
\end{itemize}

\paragraph{``Cycle'' Cuts:} To handle 2-vertex cuts of type (3), we build on ideas developed by \cite{KhullerT93,KhullerV94,JKMV24} for augmenting a cycle graph to be 3-\emph{edge}-connected. In the unweighted setting, the idea is simple. Suppose we are given some cycle $C = \{1, \dots, n\}$. Consider bidirecting all links, that is, each link $uv$ is replaced with two \emph{directed} links $(u,v)$ and $(v, u)$. The goal is, for every interval $[i, j]$, to include at least one link $(u,v)$ such that $u \notin [i, j]$ and $v \in [i, j]$: this corresponds to covering the cut $\{(i-1, i), (j, j+1)\}$, and $(u,v)$ goes from the component of $C \setminus \{(i-1, i), (j, j+1)\}$ containing vertex $1$ to the other component. It is easy to see that for $1 \leq u \leq u' < v$, the directed link $(u, v)$ covers a superset of the cuts covered by $(u', v)$, and the same holds for for $v \leq u' < u \leq n$; see Figure \ref{fig:vc_2_to_3_edge_cycle} for an example. Thus we can store at most two incoming links per vertex and obtain a 2-approximation in linear space. This can be generalized to weighted graphs using standard weight-bucketing ideas.

We employ a similar idea for covering 2-vertex cuts of type (3). Fix an S-node $x$. The main difference in this setting is that two components of $G_x \setminus \{a, b\}$ can be connected via links not in $G_x$: e.g. by a link between the two subtrees of $T$ corresponding to the two components of $G_x \setminus \{a, b\}$.
We label the vertices of $G_x$ along the cycle $\mu_0, \dots, \mu_{k}$, where $\{\mu_k, \mu_0\} = \parent(x)$. All indices are taken $\mod k+1$; this will sometimes be omitted for the ease of notation. For each \emph{virtual} edge $(\mu_i, \mu_{i+1})$, we subdivide the edge and consider a dummy vertex $\mu_{i, i+1}$. For these virtual edges, we let $T_{i, i+1}$ denote the subtree rooted at $(\mu_i, \mu_{i+1})$. We define a function $f_x$ on the vertices of $G$ that maps each vertex to its corresponding point on the cycle: $f_x(u) = \mu_{i, i+1}$ if $u \in T_{i, i+1} \setminus G_x$ and $f_x(\mu_i) = \mu_i$. See Figure \ref{fig:vc_2_to_3_algo_cycle} for an example.
We define a natural ordering $\prec_x$ on the vertices of the cycle in terms of their appearance on the cycle, so $\mu_{0, k} \prec_x \mu_0 \prec_x \mu_{0,1} \prec_x \mu_1 \prec_x \dots \prec_x \mu_k$. 
For each vertex $\mu$ on the cycle (original or dummy vertex) and each weight class $j$, we store the links $\Min_{\mu}(j) = \argmin_{(u,v) \in L}\{f_x(v): f_x(u) = \mu, w(u,v) \in [(1+\eps)^j, (1+\eps)^{j+1})\}$ and $\Max_{\mu}(j) = \argmax_{(u,v) \in L}\{f_x(v): f_x(u) = \mu, w(u,v) \in [(1+\eps)^j, (1+\eps)^{j+1})\}$. 

\begin{figure}
    \centering
    \begin{minipage}{0.4\textwidth}
        \centering
        \includegraphics[width=0.8\linewidth]{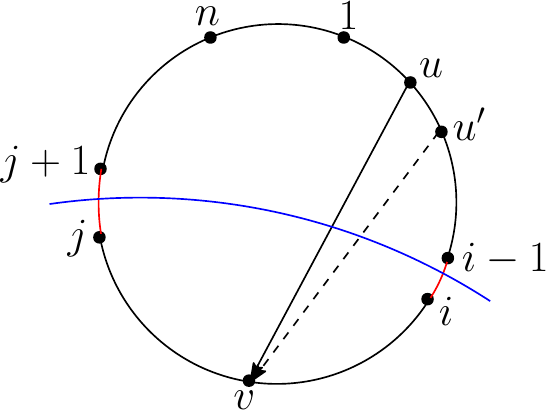}
        \caption{Example from \cite{JKMV24} to augment a cycle to be 3-edge-connected. The blue curved line indicates the cut that needs to be covered.}
        \label{fig:vc_2_to_3_edge_cycle}
    \end{minipage}%
    \hfill
    \begin{minipage}{0.5\textwidth}
        \centering
        \includegraphics[width=0.7\linewidth]{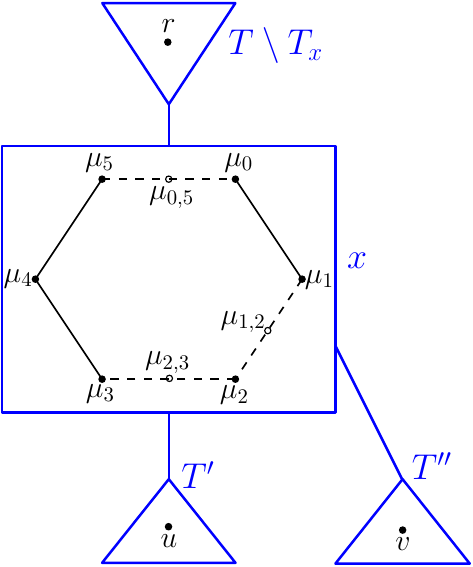}
        \caption{Example of S-node $x$ with new dummy vertices. $T'$ and $T''$ are subtrees of $x$ rooted at virtual edges $\{\mu_2, \mu_3\}$ and $\{\mu_1, \mu_2\}$ respectively. In this example, $f(u) = \mu_{2, 3}$ and $f(v) = \mu_{1, 2}$.}
        \label{fig:vc_2_to_3_algo_cycle}
    \end{minipage}
\end{figure}

\begin{algorithm}[H]
    \KwIn{An edge-minimal 2-vertex-connected graph $G = (V, E)$ with edge weights $w: E\rightarrow [1, W]$.}
    \caption{The streaming algorithm for 2-to-3 VCSS augmentation}
    \label{alg:vc_2_to_3}
    \tcc{\sc {Preprocessing}:}
    $T \gets$ SPQR tree for $G$; choose root $r \in V(T)$\\
    \For{$x \in V(T)$}{
        Initialize an empty dictionary $L_x$ \\
    }
    \For{$x \in V(T)$; $x$ is a P-node}{
        Construct $C'(x)$: for each $y \in C(x)$, contract $(\cup_{z \in T_y} V(G_z)) \setminus V(G_x)$ into supernode $y'$, add $y'$ to $C'(x)$ \\
        Define the graph $H_x$ on node set $C'(x)$ and edge set $\emptyset$
    }
    \For{$x \in V(T)$; $x$ is an S-node}{
        Create a dummy vertex for each virtual edge in $G_x$ \\
        \For{$w \in V(G_x)$ (original and dummy)}{
            Initialize empty dictionaries $\Min_w$ and $\Max_w$
        }
    }
    \tcc{\sc {In Stream}:}
    \For {$e = uv \in L$ in the stream}{
        Let $j$ be weight class such that $w(e) \in [(1+\eps)^j, (1+\eps)^{j+1})$\\
        \If{$v$ minimizes $d_T(r, LCA(h(u), \ell(v)))$}{
            Update $L_{h(u)}(j) \gets e$
        }
        \If{$u$ minimizes $d_T(r, LCA(h(v), \ell(u)))$}{
            Update $L_{h(v)}(j) \gets e$
        }
        \For{$x \in V(T)$; $x$ is a P-node}{
            Update MST $H_x$ with link $uv$ using Lemma \ref{lem:mst_streaming}
        }
        \For{$x \in V(T)$; $x$ is an S-node}{
            Update $\Min_{f_x(u)}(j)$, $\Max_{f_x(u)}(j)$, $\Min_{f_x(v)}(j)$, and $\Max_{f_x(v)}(j)$
        }
    }
    \tcc{\sc {Postprocessing}:}
    $F = (\cup_{x \in V(T)} L_x) \cup (\cup_{x \in V(T), x\text{ is P-node}} E(H_x)) \cup (\cup_{x \in V(T), x\text{ is S-node}} \cup_{\mu} \Min_\mu \cup \Max_\mu)$ \\
    \Return An optimal solution on edge set $F$
\end{algorithm}

In order to bound the space complexity of Algorithm \ref{alg:vc_2_to_3}, we use the following fact; this follows from well-known results including graph sparsification \cite{nagamochi1992linear} and ear-decompositions of 2-connected graphs. 

\begin{lemma}\label{cor:vc_2_to_3_sparse}
    Any edge-minimal 2-node-connected graph $G = (V, E)$ has at most $2n-2$ edges, where $n = |V|$.
\end{lemma}

\begin{lemma}\label{lem:vc_2_to_3_space}
    The number of links stored in Algorithm \ref{alg:vc_2_to_3} is $O(n \eps^{-1} \log W)$. 
\end{lemma}
\begin{proof}
    Let $B = O(\eps^{-1}\log W)$ be the number of weight classes.
    The set of links stored in Algorithm \ref{alg:vc_2_to_3} is exactly $F$. We bound each of the three sets of links separately. First, for each $x \in V(T)$, we store one link in $L_x$ per weight class. Thus 
    $|\cup_{x \in V(T)} L_x| \leq \sum_{x \in V(T)} B = |V(T)| \cdot B$.
    Next, we bound the number of links stored in the spanning trees. \[\left|\bigcup_{x \in V(T), x\text{ is P-node}} E(H_x)\right| \leq \sum_{x \in V(T), x\text{ is P-node}} |E(H_x)| \leq \sum_{x \in V(T)} |C(x)| - 1 \leq \sum_{x \in V(T)} \deg_T(x) = 2|E(T)|.\] 
    Finally, for each S-node $x$ and each dummy/original node $\mu$, $|\Min_\mu \cup \Max_\mu| \leq 2B$. We introduce at most one dummy node per edge of $G_x$, thus the total number of links stored for $x$ is at most $2B|V(G_x)| + 2B|E(G_x)|$. Since $G_x$ is a simple cycle, $|E(G_x)| = |V(G_x)|$; thus the number of links stored is $4B|E(G_x)|$. Combining, 
    \begin{align*}
        |F| \leq B \left[|V(T)| + 2|E(T)| + \sum_{x \in V(T)} 4|E(G_x)|\right] \leq B \left[7 \sum_{x \in V(T)} |E(G_x)|\right].
    \end{align*}
    By Lemma \ref{lem:spqr_size}, this is at most $O(B \cdot |E|)$. Corollary \ref{cor:vc_2_to_3_sparse} gives us our desired bound of $O(n \eps^{-1} \log W)$. 
\end{proof}

\subsubsection{Bounding the Approximation Ratio}

For the rest of this section, we bound the approximation ratio by proving the following lemma. 

\begin{lemma}\label{lem:vc_2_to_3_approx}
    Let $F$ be the set of links stored in Algorithm \ref{alg:vc_2_to_3}. The weight of an optimal solution to $2$-VC-CAP on $(V, F)$ is at most $(7+\eps)$ times the optimal solution to $2$-VC-CAP on $G = (V, E)$. 
\end{lemma}

Fix some optimal solution $\opt$ on $G = (V, E)$. 
To prove Lemma \ref{lem:vc_2_to_3_approx}, we provide an algorithm to construct a feasible solution $\sol \subseteq F$ such that $w(\sol) \leq (7 + \eps) w(\opt)$. Note that this algorithm is for analysis purposes only, as it requires knowledge of an optimal solution $\opt$ to the instance on $G = (V, E)$. The construction of $\sol$ from the links in $\cup_x L_x$ (links that minimize LCA distance to root) and $\cup_x E(H_x)$ (MST links) follows a similar approach to the $1$-VC-CAP setting. The key technical challenge lies in the cycle case.
This is because a single link $uv \in L$ may be responsible for increasing the connectivity of multiple S-nodes in the tree, leading to a large blow-up in the approximation factor. To circumvent this issue, we show that for each $uv \in L$, we can restrict attention to at most three S-nodes in the tree: the LCA of $\ell(u)$ and $\ell(v)$, an S-node containing $u$, and an S-node containing $v$. Details are given in Algorithm \ref{alg:vc_2_to_3_analysis} and the following analysis.

\begin{algorithm}[h]
\caption{Construction of a solution in $F$ from $\opt$.}
\label{alg:vc_2_to_3_analysis}
    $\sol \gets \emptyset$ \\
    \For{$e = (u,v) \in \opt$}{
        Let $j$ be weight class such that $w(e) \in [(1+\eps)^j, (1+\eps)^{j+1})$ \\
        $\sol \gets \sol \cup \{L_{h(u)}(j), L_{h(v)}(j)\}$\\
        \If{$x = LCA(\ell(u), \ell(v))$ is S-node}{
            \If{$f_x(u) \prec_x f_x(v)$}{
                $\sol \gets \sol \cup \{\Min_{f_x(v)}(j), \Max_{f_x(u)}(j)\}$
            }
            \If{$f_x(u) \prec_x f_x(v)$}{
                $\sol \gets \sol \cup \{\Min_{f_x(u)}(j), \Max_{f_x(v)}(j)\}$
            }
        }
        \If{$\exists$ S-node $x$ such that $u \in G_x \setminus \parent(x)$, $\ell(v) \in T \setminus T_x$}{
            $\sol \gets \sol \cup \{\Min_{u}(j)\}$
        }
        \If{$\exists$ S-node $x$ such that $v \in G_x \setminus \parent(x)$, $\ell(u) \in T \setminus T_x$}{
            $\sol \gets \sol \cup \{\Min_{v}(j)\}$
        }
    }
    \For{$x \in V(T)$, $x$ is P-node}{
        \For{$y \in C(x)$}{
            \If{$\exists e = (u,v) \in \opt$ such that $h(u) \subseteq V(T_y)$ and $\ell(v) \subseteq V(T \setminus T_x)$}{
                Mark $y$ as ``good''
            }
        }
        Construct $C''(x)$ from $C'(x)$ (as defined in Algorithm \ref{alg:vc_2_to_3}) by contracting all supernodes corresponding to a ``good'' node $y \in C(x)$ into one ``good'' supernode \\
        $H_x' \subseteq H_x \gets$ MST on $C''(x)$\\
        $\sol \gets \sol \cup E(H_x')$
    }
    \Return $\sol$
\end{algorithm}

\smallskip
We fix $\sol$ to be the set of links given by Algorithm \ref{alg:vc_2_to_3_analysis}. 

\begin{lemma}
\label{lem:vc_2_to_3_weight}
    The edge set $\sol$ has weight at most $(7 + \eps)w(\opt)$. 
\end{lemma}
\begin{proof}
    We first consider $\sol \cap \cup_x L_x$. 
    In the first part of Algorithm \ref{alg:vc_2_to_3_analysis}, for each link $uv \in \opt$, we add two links to $\sol$, each of weight at most $(1+\eps)w(uv)$. Thus $w(\sol \cap \cup_x L_x) \leq 2(1+\eps)w(\opt)$.
 
    Next, we consider $\sol \cap \cup_x E(H_x)$ and show that this has weight at most $\opt$; the reasoning is similar to the case of $1$-VC-CAP and we rewrite it here for completeness. 
    Consider a partition of $\opt$ based on the LCA of its endpoints: for each P-node $x$, let $E^*(x) = \{uv \in \opt: LCA(\ell(u), \ell(v)) = x\}$. For each node $y \in C(x)$, there must be some path in $E \cup \opt$ from the vertices of $G$ corresponding to $T_{y}$ to the vertices of $G$ corresponding to $T \setminus T_x$ \emph{without} using the two nodes in $V(G_x)$; else $V(G_x)$ would be a 2-cut separating the vertices corresponding to $T_{y}$ from the rest of the graph, contradicting feasibility of $\opt$. These paths must each have a link ``leaving'' $T_{x}$; that is, a link $e$ with one endpoint in the vertex set of $G$ corresponding to $T_{x}$ and the other in the vertex set of $G$ corresponding to $T \setminus T_{x}$. By construction, endpoint of $e$ in $T_x$ must be in a ``good'' supernode. Thus all ``bad'' supernodes are connected to at least one ``good'' supernode in $\opt$. Furthermore, a minimal path in the contracted graph from a ``bad'' supernode to a ``good'' supernode only uses links in $E^*$, since these links all must go between subtrees of $x$ and avoid $V(G_x)$. Thus in $E^*(x)$, all ``bad'' supernodes are connected to at least one ``good'' supernode. In particular, $E^*(x)$ contains a spanning tree on the $C''(x)$. Therefore 
    \[w(\sol \cap \cup_x E(H_x)) = \sum_{x \in V(T): x\text{ is P-node}} w(E(H'_x)) \leq \sum_{x \in V(T), x\text{ is P-node}} w(E^*(x)) \leq w(\opt).\]

    Finally, we bound the weight of $\sol \cap (\cup_x (\cup_\mu \Min_u \cup \Max_u))$. For each $uv \in \opt$, we consider at most 3 S-nodes
    \begin{itemize}
        \item If $LCA(\ell(u), \ell(v))$ is an S-node, we add 2 links to $\sol$ of weight at most $(1+\eps)w(uv)$ each;
        \item If $u \in G_x \setminus \parent(x)$ for some S-node $x$ and $v$ is outside the subtree rooted at $x$, then we add one link to $\sol$ of weight at most $(1+\eps)w(uv)$ -- there can be at most one such S-node, since if $u$ has a a copy in multiple S-nodes, then it must be in the parent of all but one of them;
        \item If $v \in G_x \setminus \parent(x)$ for some S-node $x$ and $u$ is outside the subtree rooted at $x$, then we add one link to $\sol$ of weight at most $(1+\eps)w(uv)$ -- there is at most one such S-node by the same argument as above.
    \end{itemize}
    Thus for each link $uv \in \opt$, we add a total weight of at most $4(1+\eps)$ to $\sol$. 

    Combining all of the above, $w(\sol) \leq [2(1+\eps) + 1 + 4(1+\eps)]w(\opt) = (7 + 6\eps)w(\opt)$. We run the algorithm with $\eps/6$ to obtain the desired approximation ratio. 
\end{proof}

\begin{lemma}
\label{lem:vc_2_to_3_feasibility}
    $(V, E \cup \sol)$ is a 3-vertex-connected graph.
\end{lemma}
\begin{proof}
    We want to show that for all $\{a, b\} \in V$, $(V, E \cup \sol) \setminus \{a,b\}$ is connected. Fix $\{a,b\} \in V$. Since $\opt$ is feasible, it suffices to show that for all $uv \in \opt$ with $u, v \notin \{a, b\}$, there exists a $u$-$v$ path in $E \cup \sol$ that does not use $\{a,b\}$.

    Fix $uv \in \opt$ and let $j$ be the weight class such that $w(uv) \in [(1+\eps)^j, (1+\eps)^{j+1})$. By Lemma \ref{lem:spqr_cuts}, there are three cases in which $G \setminus \{a,b\}$ is disconnected: (1) $ab$ is a virtual edge incident to either two R-nodes or one R-node and one S-node, (2) $\{a, b\}$ is the vertex set associated with a P-node, and (3) $a$ and $b$ are non-adjacent nodes of an S-node. We consider each of these cases separately.
    
    \paragraph{Case 1: $\boldsymbol{e = ab}$ is a virtual edge of an R-node and an R or S-node:} Let $e'=xy \in E(T)$ be the tree edge associated with virtual edge $e$. We let $x$ denote the child node and $y$ the parent; note that this implies that $\{a, b\} = \parent(x)$. Since $x$ and $y$ are both R or S nodes, $G_x$ and $G_y$ each have at least 3 vertices. Furthermore, $G_x \setminus \{a, b\}$ and $G_y \setminus \{a,b\}$ are both connected: the R-node case is clear by definition, and cycles remain connected after the deletion of two adjacent nodes.

    First, suppose $u, v$ are both in the set of vertices of $G$ associated with $T_x$. Then, since $T_x$ remains connected despite the deletion of $e'$ and $G_x \setminus \{a,b\}$ is connected, there exists a $u$-$v$ path in $E$ avoiding $\{a, b\}$. Similarly, if $u, v$ are both in the set of vertices of $G$ associated with $T \setminus T_x$, then since $T \setminus T_x$ remains connected despite the deletion of $e'$ and $G_y \setminus \{a, b\}$ is connected, there exists a $u$-$v$ path in $E$ avoiding $\{a, b\}$. 

    Thus, one of $u$ and $v$ must be in $T_x$ and the other outside $T_x$. We assume without loss of generality that $u$ is associated with $T_x$ (i.e. $h(u) \in V(T_x)$) and $v$ with $T \setminus T_x$ (i.e. $\ell(v) \in V(T \setminus T_x)$). Let $(u', u'') = L_{h(u)}(j)$ where $h(u') = h(u)$; this must be in $\sol$ since $uv \in \opt$. By construction, $d_T(r, LCA(h(u), \ell(u''))) \leq d_T(r, LCA(h(u), \ell(v)))$ -- in particular, $LCA(h(u), \ell(u'')) \in T \setminus T_x$. See Figure \ref{fig:vc_2_to_3_case1} for reference.
    
    This gives us the following $u$-$v$ path in $\sol \cup E$ without using $\{a, b\}$. 
    \begin{itemize}
        \item $u \to u'$: since $u$ and $u'$ are both in the vertex set of $G_{h(u)}$, and all tree nodes stay connected despite the deletion of $\{a, b\}$, there is a path from $u$ to $u'$ in $E$ without $\{a, b\}$. Note that $u' \notin \{a, b\}$, since $h(u) = h(u') \in T_x$, and $h(a) = h(b) = y$.
        \item $u' \to u''$: as described above, this link is contained in $\sol$. Note that $u'' \notin \{a, b\}$: $LCA(h(u), \ell(u'')) \in T \setminus T_x$ implies that $\ell(u'') \in T \setminus T_x$ (since $h(u) \in T_x$), and $\ell(a) = \ell(b) = x$. 
        \item $u'' \to v$: since $G_y \setminus \{a, b\}$ is connected and $T \setminus T_x$ remains connected despite the deletion of $e'$, there is a $u''$-$v$ path in $E$ avoiding $\{a, b\}$. 
    \end{itemize}

    \begin{figure}
        \centering
        \includegraphics[scale=0.6]{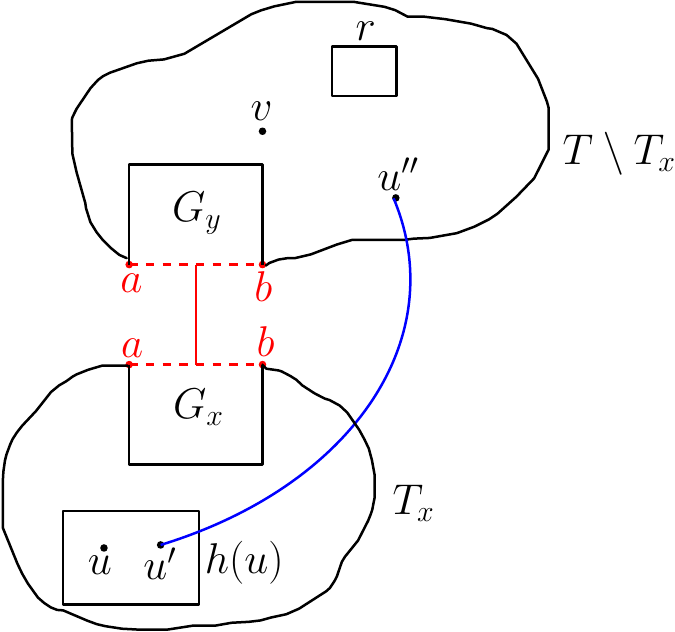}
        \caption{Case 1 example. Tree nodes are shown with boxes, while graph nodes are given by dots. Note that while the figure shows $h(u)$ and $x$ to be separate, it is possible that $h(u) = x$. Similarly, it is possible that $v, u'' \in G_y$, and possible that $y = r$.}
        \label{fig:vc_2_to_3_case1}
    \end{figure}

    \paragraph{Case 2: $\boldsymbol{V(G_x) = \{a,b\}}$ for a P-node $\boldsymbol{x}$:} Since no two P-nodes are adjacent to each other, all neighbors of $x$ in $T$ are R or S nodes. In particular, this means that for all $y \in V(T) \setminus \{x\}$, $G_y \setminus \{a, b\}$ is connected. Therefore, if $u, v$ are in the same component of $T \setminus \{x\}$, then there is a $u$-$v$ path in $E \setminus \{a, b\}$. Thus we assume $u$ and $v$ are in distinct components of $T \setminus \{x\}$. 
    \begin{itemize}[leftmargin=*]
        \item \textbf{Case 2a: $\boldsymbol{h(u), h(v)}$ are both in $\boldsymbol{T_x}$:} Let $T(u)$ and $T(v)$ be the subtrees of $T_x$ containing $h(u)$ and $h(v)$ respectively. Suppose either of the subtrees are rooted at nodes that are not ``good'' (as defined in Algorithm \ref{alg:vc_2_to_3_analysis}). Then, there must be two distinct supernodes $z, z'$ in $C''(x)$ such that all vertices associated with $T(u)$ are contracted into $z$ and all vertices associated with $T(v)$ are contracted into $z'$. Thus $T(u)$ and $T(v)$ must be connected by $H'_x$. 
        
        Therefore, we assume both subtrees are good. 
        We will show that $\sol$ contains a link $e_u$ between $T(u)$ and $T \setminus T_x$ and a link $e_v$ between $T(v)$ and $T \setminus T_x$ such that neither $e_u$ nor $e_v$ are incident to $\{a, b\}$. By the above discussion, since all components of $T \setminus x$ (and thus their corresponding vertices in $G$) remain connected despite the removal of $\{a, b\}$, this suffices to show that there exists a $u$-$v$ path in $E \cup \sol$ avoiding $\{a, b\}$. 

        Since $T(u)$ is good, there exists a link $(u_1, u_2) \in \opt$ with $h(u_1) \subseteq T(u)$ and $\ell(u_2) \subseteq T \setminus T_x$. In particular, this means that $LCA(h(u_1), \ell(u_2)) \in T \setminus T_x$. Let $j'$ be the weight class of $(u_1, u_2)$, and let $(u_3, u_4) = L_{h(u_1)}(j')$ with $h(u_1) = h(u_3)$. This must be in $\sol$ since $(u_1, u_2) \in \opt$. By construction, $d_T(r, LCA(h(u_1), \ell(u_2))) \leq d_T(r, LCA(h(u_1), \ell(u_4)))$ -- in particular, $LCA(h(u_1), \ell(u_4)) \in T \setminus T_x$, so $\ell(u_4) \in T \setminus T_x$. Note that $u_3 \neq \{a, b\}$ since $h(u_3) = h(u_1) \in T(u)$, and $h(a) = h(b) \in T \setminus T_x$. Furthermore, $u_4 \neq \{a, b\}$, since $\ell(u_4) \in T \setminus T_x$, but $\ell(a), \ell(b) \in T_x$. Thus $(u_3, u_4)$ is a link in $\sol$ between $T(u)$ and $T \setminus T_x$ that is not incident to $a$ or $b$, as desired. The analysis for $v$ is analogous. See Figure \ref{fig:vc_2_to_3_case2} for reference.
        \item \textbf{Case 2b: One of $\boldsymbol{u}$ and $\boldsymbol{v}$ is in $\boldsymbol{T_x}$ and the other is in $\boldsymbol{T \setminus T_x}$:} Without loss of generality, suppose $h(u) \in T_x$ and $\ell(v) \in T \setminus T_x$. In this case, we use the same argument as Case 1: let $(u', u'') = L_{h(u)}(j)$ where $h(u') = h(u)$; this must be in $\sol$ since $uv \in \opt$. Clearly $u' \notin \{a, b\}$, since $h(u') = h(u) \in T(u)$ and $h(a), h(b) \in T \setminus T_x$. By the same reasoning as in Case 1, $\ell(u'') \in T \setminus T_x$, so $u'' \notin \{a, b\}$, since $\ell(a), \ell(b) \in T_x$. Thus $E \cup \opt$ contains a $u$-$v$ path avoiding $\{a, b\}$ via the link $(u', u'')$. 
    \end{itemize}

    \begin{figure}
        \centering
        \begin{subfigure}{0.35\textwidth}
            \includegraphics[scale=0.5]{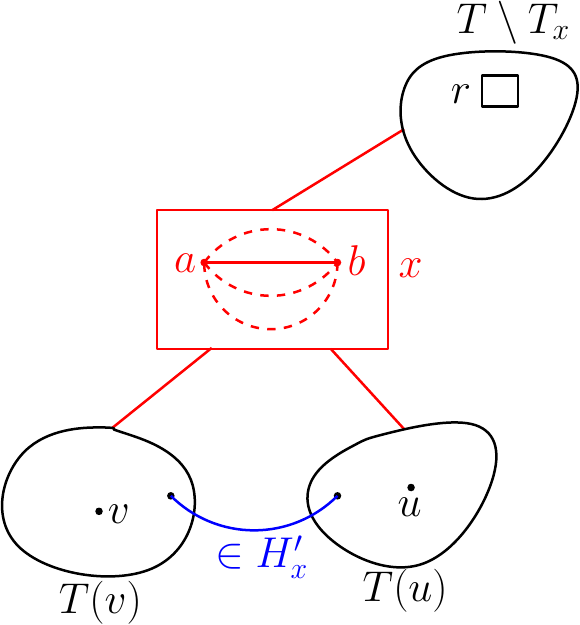}
            \caption{Either $T(v)$ or $T(u)$ are not good; then they are connected via the MST $H_x'$.}
        \end{subfigure}
        \hfill
        \begin{subfigure}{0.6\textwidth}
            \includegraphics[scale=0.65]{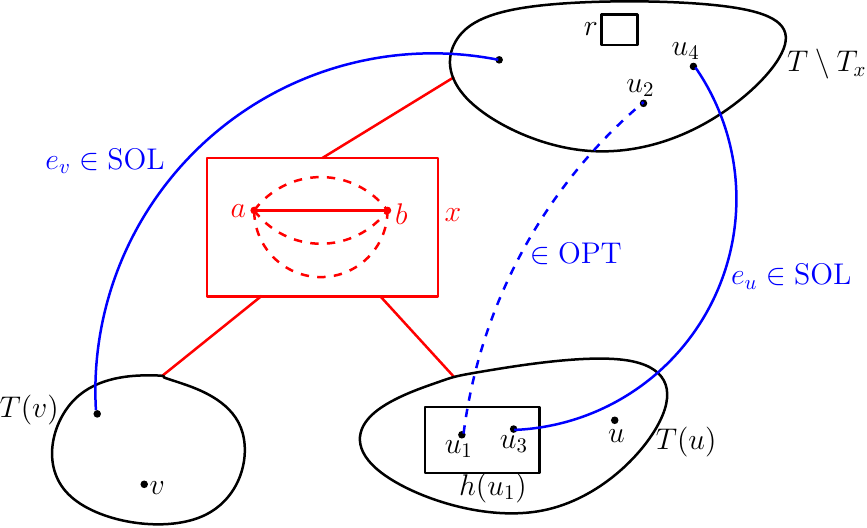}
            \caption{Both $T(u)$ and $T(v)$ are good. Note that $u$ may be in $h(u_1)$, though this is not necessary as demonstrated in the figure.}
        \end{subfigure}
        \caption{Case 2a example. Tree nodes are shown with boxes, while graph vertices are given by dots. Here $x$ has two subtrees $T(u)$ and $T(v)$.}
        \label{fig:vc_2_to_3_case2}
    \end{figure}

    \paragraph{Case 3: $\boldsymbol{a}$ and $\boldsymbol{b}$ are non-adjacent nodes of $\boldsymbol{G_x}$ for an S-node $\boldsymbol{x}$:} Let $C_1$ and $C_2$ be the two components of the cycle $G_x$ formed by the removal of $a$ and $b$. We assume that $u$ and $v$ are separated in $G$ by the removal of $\{a, b\}$ -- this implies that $x$ must be on the tree path from $u$ to $v$ (this includes the case that $u$ and $v$ are in the cycle $G_x$). Note that since $a$ and $b$ are non-adjacent, all dummy nodes remain intact. We overload notation and write $G_x$ to include dummy nodes. For ease of notation we write $f := f_x$, $\prec := \prec_x$ for this section. 
    \begin{itemize}[leftmargin=*]
        \item \textbf{Case 3a: $\boldsymbol{\text{LCA}(\ell(u), \ell(v)) = x}$:} Notice that $u$ and $f(u)$ (as well as $v$ and $f(v)$) remain connected despite the deletion of $\{a, b\}$. Thus $f(u)$ and $f(v)$ must be in distinct components of $G_x \setminus \{a, b\}$; without loss of generality suppose $f(u) \in C_1$ and $f(v) \in C_2$. 
        \begin{itemize}
            \item If $f(u) \prec f(v)$, then let $(v', v'') = \Min_{f(v)}(j)$ with $f(v') = f(v)$. Clearly, $v', f(v') \notin \{a, b\}$, since $f(v) = f(v')$, so either $v' = v = f(v)$ or $v$ and $v'$ are in the same subtree rooted at $x$ and $f(v')$ is a dummy node. Furthermore, note that $(v', v'') \in \sol$ since $uv \in \opt$ and $LCA(\ell(u), \ell(v)) = x$. By construction, $f(v'') \preceq f(u)$ and is therefore also in $C_1$. Since $f(u) \in C_1$, $f(u) \preceq \min(a, b)$ (where $\min$ is with respect to $\prec$), thus $f(v'') \notin \{a, b\}$ and by extension $v'' \notin \{a, b\}$. Thus there exists a $u$-$v$ path: $u \to f(u) \to f(v'')$ (since $C_1$ remains connected), $f(v'') \to v'' \to v' \to f(v) \to v$. See Figure \ref{fig:vc_2_to_3_case3a} for reference.
            \item The case with $f(v) \prec f(u)$ is similar; in this case we let $(v', v'') = \Max_{f(v)}(j)$. It is easy to see that by the same argument as above, $v', v'', f(v'')$ are all not contained in $\{a, b\}$, and thus there exists a $u$-$v$ path via the $(v', v'')$ link in $\sol$. 
        \end{itemize}
        \item \textbf{Case 3b: $\boldsymbol{\text{LCA}(\ell(u), \ell(v)) \neq x}$: } Since we assume that $x$ is on the tree path between $u$ and $v$, we can assume without loss of generality that $\ell(u) \in T_x$ and $\ell(v) \in T \setminus T_x$. Note that this implies that $v \notin G_x$ and $f(v)$ is the dummy node corresponding to the parent edge of $x$. 
        \begin{itemize}
            \item Suppose $u \in G_x \setminus \parent(x)$. Then $\Min_u(j) = (u, v')$ where $f(v') = f(v)$, so $\ell(v') \in T \setminus T_x$. It is easy to see that $T \setminus T_x$ remains connected despite the removal of $\{a, b\}$. Furthermore, $(u, v') \in \sol$, since $(u, v) \in \opt$, $u \in G_x \setminus \parent(x)$, and $\ell(v) \in T \setminus T_x$. Finally, $v' \notin \{a, b\}$ since $v' \notin G_x$. Thus there exists a $u$-$v$ path in $E \cup \sol$ avoiding $\{a, b\}$ via the link $\{u, v'\}$. See Figure \ref{fig:vc_2_to_3_case3b} for reference.
            \item Suppose $u \in \parent(x)$. Then, since $a$ and $b$ are non-adjacent nodes of $G_x$, $u$ is connected to $T \setminus T_x$ despite the removal of $\{a, b\}$, so $E$ contains a $u$-$v$ path avoiding $\{a, b\}$.
            \item Suppose $u \notin G_x$. Then we can use the same argument as in Case 1. let $(u', u'') = L_{h(u)}(j)$ where $h(u') = h(u)$; this must be in $\sol$ since $(u,v) \in \opt$. Furthermore, $u' \notin \{a, b\}$, since $h(u') = h(u)$ which is in a strict subtree of $T_x$, but $h(a), h(b)$ are at least as close to the root as $x$. By the same reasoning as Case 1, $\ell(u'') \in T \setminus T_x$, so $u'' \notin \{a, b\}$, since $\ell(a), \ell(b) \in T_x$. Thus $E \cup \opt$ contains a $u$-$v$ path avoiding $\{a, b\}$ via the link $(u', u'')$. 
        \end{itemize}
    \end{itemize}
    \begin{figure}
        \centering
        \begin{subfigure}{0.45\textwidth}
            \centering
            \includegraphics[scale=0.6]{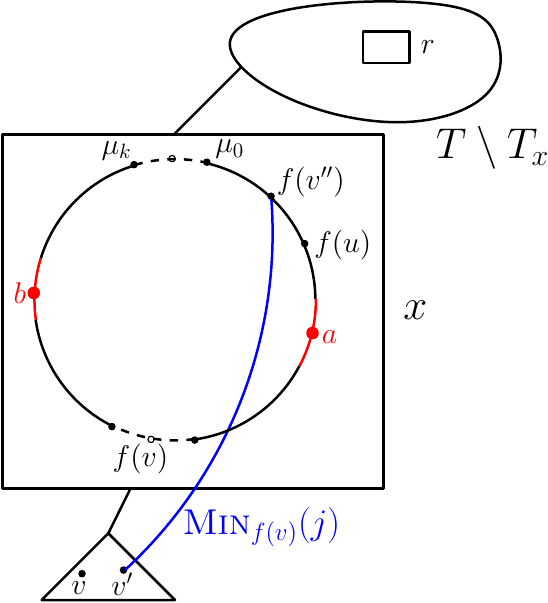}
            \caption{Case 3a: Here $f(u) \prec f(v)$. In this example, $f(v'') = v''$ and $f(u) = u$, but $v, v'$ are not on the cycle and thus $f(v)$ is a dummy node.}
            \label{fig:vc_2_to_3_case3a}
        \end{subfigure}
        \hfill
        \begin{subfigure}{0.45\textwidth}
            \centering
            \includegraphics[scale=0.6]{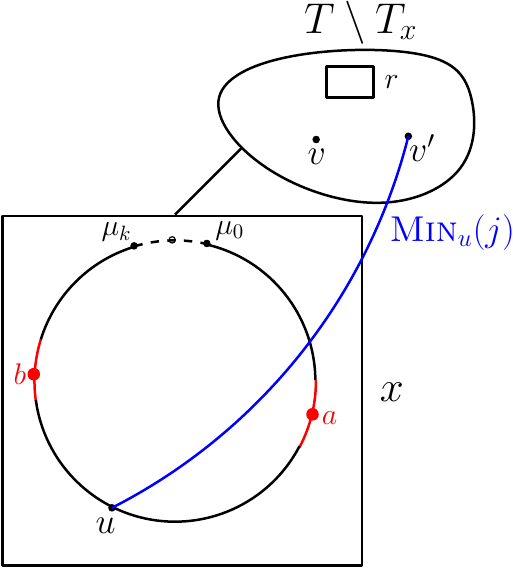}
            \caption{Case 3b: Here $u \in G_x \setminus \parent(x)$.}
            \label{fig:vc_2_to_3_case3b}
        \end{subfigure}
        \caption{Case 3 example. Tree nodes are shown with boxes, while graph vertices are given by dots.}
        \label{fig:vc_2_to_3_case3}
    \end{figure}
\end{proof}

By Lemmas \ref{lem:vc_2_to_3_feasibility} and \ref{lem:vc_2_to_3_weight}, Algorithm \ref{alg:vc_2_to_3_analysis} provides a feasible solution to $2$-VC-CAP on $(V, F)$ with weight at most $(7 + \eps) w(\opt)$; this concludes the proof of Lemma \ref{lem:vc_2_to_3_approx}. This, combined with Lemma \ref{lem:vc_2_to_3_space}, concludes the proof of Theorem \ref{thm:vc_2_to_3_main}.

\section{Lower Bounds for Streaming Network Design}\label{sec:vc-sndp-lb}

We describe a lower bound for the vertex-connectivity tree-augmentation problem (VC-TAP), i.e., $1$-VC-CAP. 

\begin{theorem}\label{thm:vc-tap-spanner-lb}
Consider the unweighted VC-TAP where both $E$ and $L$ arrives as a stream in an arbitrary order. Any single-pass algorithm returning a better than $(2t+1)$-approximation requires $\Omega(n^{1+1/t})$ space.
\end{theorem}
\begin{proof}
Let $G = (V, E)$ be a fixed graph on $n$ vertices with girth strictly larger than $2t + 1$. Consider the INDEX problem: Alice has a bit string from $\{0,1\}^{E(G)}$, representing a subgraph $G' \subseteq G$. Alice then sends a message to Bob, who must recover the $i$-th bit of the string for a given index $i$, or equivalently, determine whether $(u, v) \in G'$ for a specified edge $(u, v) \in G$. It was shown by Miltersen { \em et al.}~\cite{jcss/MiltersenNSW98} that any bounded-error randomized protocol for INDEX requires a message of size $\Omega(|E(G)|)$ bits. Note that the graph $G$ and its edges are known to both Alice and Bob.

\smallskip
We now use the streaming algorithm $\mathcal{A}$ for VC-TAP to design a protocol for the INDEX problem. Alice and Bob jointly construct a TAP instance $(E, L)$ for $\mathcal{A}$ as follows:
\begin{itemize}
    \item First, Alice sets $L := E(G')$ and feeds it into $\mathcal{A}$. She then sends the memory contents of $\mathcal{A}$ to Bob. To determine whether $(u, v) \in G'$, Bob constructs a chain $E := \{(x_1, x_2), (x_2, x_3), \dots, (x_{|V|-1}, x_{|V|})\}$ and feeds $E$ to $\mathcal{A}$, where $x_1 := u$, $x_{|V|} := v$, and the intermediate vertices $\{x_2, \dots, x_{|V|-1}\} = V \setminus \{u, v\}$ are ordered so that $d_{H}(u, x_{j-1}) \leq d_{H}(u, x_{j})$ for $2 \leq j \leq |V| - 1$. Here, the graph $H$ is defined as $G$ with the edge $(u, v)$ removed.
    \item Bob then determines that $(u, v) \in G'$ if and only if $\mathcal{A}$ returns an approximate solution with cost less than $2t + 1$ for the instance $(E, L)$.
\end{itemize}
Next, we prove the correctness of the described protocol that uses $\mathcal{A}$.

\begin{itemize}
  \item {\bf If $(u,v)\in E(G')$:} In this case, the optimal solution for augmenting the chain $E$ is to add the single edge $(u, v) \in L = E(G')$, completing a spanning cycle. Hence, $\mathcal{A}$ will report an approximate solution of cost at most $2t+1$ and Bob can correctly decide $(u,v)\in E(G')$.

\item {\bf If $(u,v)\notin E(G')$:} To show that Bob can correctly decide $(u,v)\notin E(G')$, it suffices to show that any feasible augmentation set $S=\{(x_{i_1},x_{j_1}),(x_{i_2},x_{j_2}),\dots,$ $(x_{i_s},x_{j_s})\}\subseteq L=E(G')$ has size at least $2t+1$.

We assume $i_k<j_k$ for all $1\le k\le s$. 
Since $\{(x_1,x_2),(x_2,x_3),\dots,$ $(x_{|V|-1},x_{|V|})\}\cup S$ is $2$-vertex-connected, the intervals $[i_1+1,j_1-1],\dots,[i_s+1,j_s-1]$ covers all $\{2,\dots,|V|-1\}$. This is because, once we add an edge $(i_k, j_k)$, removing any vertex in $[i_k + 1, j_k - 1]$ does not disconnect the graph, as the connectivity is preserved by the edge $(i_k, j_k)$. We can assume, without loss of generality, that $1 = i_1 < i_2 < \dots < i_s$ and $j_1 < \dots < j_s = |V|$, with $j_{k-1} > i_k$, by keeping a minimal feasible subset of $S$. This is because if $i_{k-1} < i_k$ and $j_{k-1} > j_k$, then removing $(i_k, j_k)$ still leaves all cuts covered by $(i_{k-1}, j_{k-1})$. Moreover, if $j_{k-1} \le i_k$, the removing any nodes in $[j_{k-1},i_k]$ will leave the graph $\{(x_1,x_2),(x_2,x_3),\dots, (x_{|V|-1},x_{|V|})\}\cup S$ disconnected.

Note that $S\subseteq E(G') \subseteq E(G) \setminus \{(u,v)\} = E(H)$.
Now we inductively prove for every $1\le k \le s$ that $d_H(u,x_{j_k})\le k$.   The base case $k=1$ is immediate: $d_H(u,x_{j_1}) = d_H(x_{i_1},x_{j_1}) = 1$. For the inductive step $2\le k\le s$, we have
\begin{align*}
    d_H(u,x_{j_k}) 
    &\le d_H(u,x_{i_{k}})+d_H(x_{i_k},x_{j_k})\\
    & = d_H(u,x_{i_{k}}) + 1\\
    & \le d_H(u,x_{j_{k-1}}) +1  &&\hspace{-1.1cm}\rhd\text{by $i_k\le j_{k-1}<|V|$ and monotonicity of $d_H(u,x_{*})$}\\
    & \le k &&\hspace{-1.1cm}\rhd\text{by induction hypothesis}
\end{align*}
Hence, this shows that $d_H(u,v) = d_{H}(u,x_{j_s})\le s$. So $G = H \cup \{(u,v)\}$ contains a cycle of length at most $s+1$.  Since $G$ has girth at least $2t+2$, we conclude $s \ge 2t+1$.
\end{itemize}
\end{proof}
Theorem~\ref{thm:vc-tap-spanner-lb} immediately implies the following for $k$-VC-CAP. 
\begin{corollary}\label{cor:vc-cap-lb}
Any algorithm approximating $k$-VC-CAP with a factor better than $2t+1$ in fully streaming requires $\Omega(n^{1+1/t})$ space.
\end{corollary}

Moreover, by Theorem~\ref{thm:vc-tap-spanner-lb} and the fact that any feasible solution for $k$-VCSS and VC-SNDP (in general) has size $\Omega(nk)$, the following corollaries hold.

\begin{corollary}\label{cor:vc-vcss-lb}
Any algorithm approximating $k$-VCSS with a factor better than $2t+1$, in insertion-only streams, requires $\Omega(nk + n^{1+1/t})$ space.
\end{corollary}

\begin{corollary}\label{cor:vc-sndp-lb}
Any algorithm approximating VC-SNDP with maximum connectivity requirement $k$ with a factor better than $2t+1$, in insertion-only streams, requires $\Omega(nk + n^{1+1/t})$ space.
\end{corollary}

\section*{Acknowledgment}

This work was conducted in part while Chandra Chekuri, Sepideh Mahabadi and Ali Vakilian were visitors at the Simons Institute for the Theory of Computing as part of the Data Structures and Optimization for Fast Algorithms program. 
The authors thank Ce Jin for his contributions during the early stage of the project. The authors also thank an anonymous reviewer for pointing out polynomial-time computability issues in the greedy fault-tolerant spanner construction, and for pointers to related literature on efficient constructions and improved bounds for fault-tolerant spanners. 

\bibliographystyle{alpha}
\bibliography{main}

\end{document}